\documentclass{ieeeconf}

\usepackage{amsmath,amssymb,amsfonts}
\usepackage{algorithmic}
\usepackage{graphicx}
\usepackage{algorithm,algorithmic}
\usepackage{hyperref, mathtools}
\hypersetup{hidelinks=true}
\usepackage{textcomp}
\def\BibTeX{{\rm B\kern-.05em{\sc i\kern-.025em b}\kern-.08em
    T\kern-.1667em\lower.7ex\hbox{E}\kern-.125emX}}

\newcommand{\Left}[0]{{\mathrm{L}}}
\newcommand{\Right}[0]{{\mathrm{R}}}

\DeclareMathOperator{\Ad}{Ad}

\DeclareMathOperator{\Diff}{Diff}
\DeclareMathOperator{\Aut}{Aut}

\DeclareMathOperator{\id}{id}

\newcommand{\aut}[0]{{\mathfrak{aut}}}
\newcommand{\X}[0]{{\mathfrak{X}}}

\newcommand{\diff}[0]{{\mathrm{d}}}
\newcommand{\ddt}[0]{\tfrac{\diff}{\diff t}}

\newcommand{\T}{\mathrm{T}}
\newcommand{\V}{\mathrm{V}}

\newcommand{\R}{\mathbb{R}}
\let\H\relax
\newcommand{\H}{\mathrm{H}}
\newcommand{\ver}{\mathrm{ver}}
\newcommand{\hor}{\mathrm{hor}}

\newcommand{\Cn}{\mathrm{Cn}}
\newcommand{\cn}{\mathrm{cn}}

\usepackage{amsthm}
\usepackage{thmtools}
\renewcommand\thmcontinues[1]{\textit{continued}}

\declaretheoremstyle[notefont=\normalfont\itshape,bodyfont=\normalfont]{normaltext}

\newtheorem{theorem}{Theorem}
\newtheorem{corollary}{Corollary}
\declaretheorem[name=Definition,style=normaltext]{definition}

\newtheorem{proposition}{Proposition}
\newtheorem{lemma}{Lemma}
\newtheorem{fact}{Fact}

\newtheorem{assumption}{Assumption}

\usepackage{enumerate}

\usepackage{comment}

\usepackage[compress]{cite} 

\usepackage{ulem}

\usepackage{tikz}
\usetikzlibrary{positioning,arrows.meta}
\usetikzlibrary{calc}
\usetikzlibrary{shapes.geometric}
\usetikzlibrary{backgrounds}

\usepackage{microtype}

\usepackage{tikzscale}

\tikzset{
    auto,
    >=latex,
    block/.style={
        draw,
        semithick,
        rectangle,
        minimum height=.7cm,
        minimum width=.7cm,
        align=center
    },
    integrator/.style={
        draw,
        ellipse,
        semithick,
        minimum height=.75cm,
        minimum width=.5cm,
        inner sep=0pt,
    },
    sum/.style={
        draw,
        semithick,
        circle,
        minimum size=0.6cm,
        inner sep=0pt, 
        append after command={
        \pgfextra{\let\LN\tikzlastnode}
            node[above=-.04cm]     at (\LN.center) {\tiny$+$}
            node[below=-.04cm]    at (\LN.center) {\tiny$+$}
        }
    },
    arrow/.style={
        -{latex[length=2mm]},
		semithick
    },
    splice/.style={
        circle, 
		fill,
        black,
        minimum size=.1cm,
        inner sep = 0pt,
        node contents = {}
    },
    node distance=1.5cm and 2cm
}

\definecolor{ieee_blue}{rgb}{0,0.263,0.576}
\definecolor{ieee_green}{rgb}{0,0.576,0.263}

\usepackage{amsmath}
\usetikzlibrary{arrows.meta,calc}

\newif\ifshowgravity
\showgravitytrue           %
\newif\ifshowinertialframe
\showinertialframetrue
\newif\ifshowequations
\showequationsfalse

\definecolor{FrameBlue}{rgb}{0,0.263,0.576}
\definecolor{ToolOrange}{rgb}{0.576,0,0.063}
\definecolor{RotorGray}{RGB}{218,220,223}
\definecolor{BodyGray}{RGB}{120,126,132}
\definecolor{LinkGray}{RGB}{155,159,164}

\pgfdeclarelayer{back}
\pgfdeclarelayer{front}
\pgfsetlayers{back,main,front}

\tikzset{
  arm outline/.style={draw=black!72,line width=3.0pt},
  arm/.style={draw=BodyGray!78,line width=2.05pt},
  boom outline/.style={draw=black!75,line width=5.6pt},
  boom/.style={draw=LinkGray,line width=4.2pt},
  rotor disk/.style={draw=none,fill=RotorGray,line width=.55pt},
  axis/.style={-{Latex[length=3mm,width=2.5mm]},
                    line width=1.5pt},
  force/.style={-{Latex[length=3mm,width=2.5mm]},
                    line width=1.5pt},
  offset guide/.style={draw=ToolOrange,line width=1pt,
                       dash pattern=on 2pt off 1.25pt},
  tool geometry/.style={draw=black!85,line width=1.25pt},
  label backing/.style={inner sep=.35pt},
  equation box/.style={draw=black!25,fill=white,rounded corners=2pt,
                       inner xsep=6pt,inner ysep=4pt,align=left}
}

\tikzset{
  rotor spin/.style={
    -{Latex[length=2.6mm,width=3.0mm]},
    draw=black,
    line width=1.1pt
  }
}

\newcommand{\DrawRotor}[4]{%
  \begin{scope}[shift={(#1)},rotate=#2]
    \draw[rotor disk] (0,0) ellipse[x radius=#3,y radius=#4];
  \end{scope}%
}

\newcommand{\DrawRotorSpin}[6]{%
  \begin{scope}[shift={(#1)},rotate=#2]
    \draw[rotor spin]
      ({#3*cos(#5)},{#4*sin(#5)})
      arc[
        start angle=#5,
        end angle=#6,
        x radius=#3,
        y radius=#4
      ];
  \end{scope}%
}

\def\comscale{1.3}
\newcommand{\DrawQuarteredHub}[1]{%
  \begin{scope}[shift={(#1)}]
    \fill[white] (0,0) circle[radius=.125*\comscale];
    \fill[black] (0,0) -- (0:.125*\comscale) arc[start angle=0,end angle=90,
      radius=.125*\comscale] -- cycle;
    \fill[black] (0,0) -- (180:.125*\comscale) arc[start angle=180,end angle=270,
      radius=.125*\comscale] -- cycle;
    \draw[black,line width=.55pt] (0,0) circle[radius=.125*\comscale];
  \end{scope}%
}

\begin{document}

\title{\bfseries A Weak Notion of Symmetry for Control Systems}
\author{Jake Welde, Riley Link, and Pieter van Goor
\thanks{J. Welde is with the Sibley School of Mechanical and Aerospace Engineering, Cornell University, Ithaca, NY,  USA,
        {\tt\small jakewelde@cornell.edu}.}
\thanks{R. Link is with the Center for Applied Mathematics, Cornell University, Ithaca, NY,  USA,
        {\tt\small rjl333@cornell.edu}.}
\thanks{P. van Goor is with the School of Aerospace, Mechanical, and Mechatronic Engineering, The University of Sydney, NSW 2006, Australia, 
        {\tt\small pieter.vangoor@sydney.edu.au}.}
}

\maketitle

\begin{abstract}
    Symmetry (or invariance) is a powerful 
    structural property that enables
    efficient, effective solutions for estimation and control. 
    However, the constraints imposed  
    on a system's dynamics
    by classical invariance 
    make symmetry a very rigid property, which may be broken by external forces or confined to only a portion of the overall  system. 
    Seeking greater flexibility, 
    this work introduces a novel relaxed notion of symmetry, termed ``weak invariance'', 
    in which the non-symmetric part of the dynamics (the ``residual'') can be captured entirely by another control system evolving on the symmetry group.
    Weakly invariant systems are strictly more general than 
    classical
    invariant systems, but they nonetheless enjoy many similar favorable properties.
    In particular, we prove that any weakly invariant system admits a cascade decomposition in which the driven subsystem is group affine, showing that weak symmetry generalizes not only 
    classical
    symmetry, but also the (thus far distinct) class of group affine systems. 
    We also show that a weak symmetry with autonomous residual can be factored out of the system's error dynamics, enabling yet a greater reduction of dimensionality as compared to  
    classical symmetries.
    Finally, we study the example of
    an aerial vehicle under the influence of gravity, for which we propose a nine-dimensional weak symmetry (strictly containing the system's familiar four-dimensional 
    classical
    symmetry). 
    Weak invariance thus generalizes classical symmetry while also preserving key structural properties, thereby laying a foundation for  more flexible 
    methods of symmetry-informed control.
\end{abstract}

\section{Introduction}
\label{sec:introduction}

Symmetry
is a powerful property 
exhibited by many control systems of practical significance, in which a family of transformations (parameterized by elements of a Lie group) leaves the dynamics invariant.
The study of symmetry in dynamical systems began in the mechanics literature for Lagrangian systems \cite{poincare1901_euler_poincare_equation, Noether1918, arnold1966geometrie}, and early work in the control community explored controllability for invariant systems on Lie groups \cite{brockett1972systemtheory, JurdjevicSussmann1972Control}. 
Following those seminal works, invariance was studied in
port-Hamiltonian systems \cite{vanderschaft1981symmetries}
and general nonlinear systems  \cite{Grizzle1985} evolving on arbitrary smooth manifolds, 
revealing the deep implications of symmetry on system structure.

\begin{figure}[t]
    \centering
    \vspace{6pt}
\begin{center}
\resizebox{.95\columnwidth}{!}{
    \begin{tikzpicture}[
    line width=1pt,
    every node/.style={
        font=\sffamily\footnotesize,
        align=center
    }
]

\draw[line width=1.5pt, ToolOrange] (0,0) ellipse [x radius=3.5cm, y radius=3.5cm];
\node[ToolOrange,
    text width=5cm,
] at (0,2.25) { \bfseries\normalsize Weakly $\mathbf{\Phi}$-Invariant \\[2pt] Systems};

\node[black,
    text width=2cm,
] at (0,-.45) { Left-Invariant Systems};

\draw[line width=1.5pt, ieee_green] (1.0,-.4) ellipse [x radius=2.05cm, y radius=2.05cm];
\node[ieee_green, 
    text width=3cm,
] at (2.05,-.45) {\bfseries Group Affine \\ Systems};

\draw[line width=1.5pt, ieee_blue] (-1.0,-.4) ellipse [x radius=2.05cm, y radius=2.05cm];
\node[ieee_blue,
    text width=2cm,
] at (-2.05,-.45) {\bfseries $\mathbf{\Phi}$-Invariant Systems};

\draw[line width=1pt] (-4,-4) rectangle (4,4);
\node[text width=2cm] at (3.1,-3.2) {General Nonlinear Systems};

\end{tikzpicture}
}
\end{center}
\caption{In this work, we propose a relaxed notion of symmetry for control systems, which generalizes both classically $\Phi$-invariant control systems on arbitrary smooth manifolds and group affine systems on Lie groups. We show that weakly invariant systems enjoy many favorable properties analogous to those exhibited under more familiar notions of symmetry.
}
\label{fig:venn_diagram}
\end{figure}
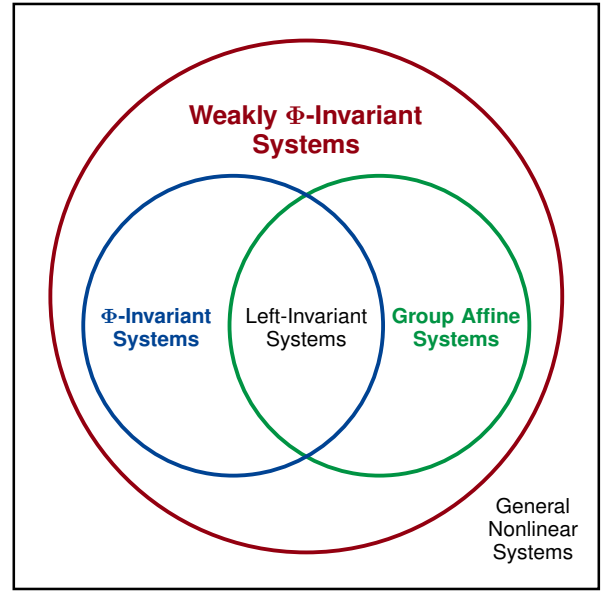

A vast array of 
methods
in control, estimation, and learning have leveraged such structural insights to enable substantial gains  in efficiency, generalization, or performance. 
Symmetry has played a significant role
in the design of tracking controllers \cite{martin2004invariant, Welde2024}, in the control of multi-agent systems \cite{Vasile2018}, and in understanding the relationship between internal shape changes and overall locomotion in biological organisms and robotic systems \cite{Shapere_Wilczek_1989,BKMM}.
Simultaneously, symmetry has been exploited in observer design; for invariant systems on Lie groups, the error dynamics governing relative motion are state-independent \cite{Bonnabel2008tac, Bonnabel2009tac}, and equivariant systems evolving on homogeneous spaces enjoy similar properties \cite{mahony2022ObserverDesignNonlinear}, facilitating the design of observers with strong convergence properties. 
Meanwhile, in data-driven control, the inductive bias of symmetry has enabled distributed learning of safe controllers \cite{Bousias2025}, efficient planning algorithms \cite{zhao2023symplan}, and improved generalization in reinforcement learning for robotic manipulation \cite{Wang2022}.
In many cases, symmetry 
can be used to establish a rigorous 
correspondence  between the original control system and a lower-dimensional abstraction thereof (\textit{e.g.}, via
a bisimulation relation \cite{vanderSchaft2004bisimulation} %
or a homomorphism of Markov decision processes \cite{Panangaden2024}), benefiting sample efficiency and enabling generalization by construction \cite{welde2025leveragingsymmetryacceleratelearning}.

Although the classical notion of symmetry provides a powerful framework for analysis and design, the strict requirements it imposes on a system's dynamics limit its broader applicability. 
A number of authors have sought to overcome this by introducing other forms of symmetry with similar structural characteristics.
An ``equivariant system'' \cite{Bonnabel2009tac,mahony2022ObserverDesignNonlinear} exhibits symmetry with respect to a transitive action on the state space along with another action of the same group on the input space; remarkably, the dependence of the error dynamics of an equivariant system on the current state can be captured entirely by the transformation of the input \cite{mahony2022ObserverDesignNonlinear}, which has been exploited in a large number of observer designs \cite{Barrau2017,vangoor2023EquivariantFilterEqF}. Relatedly, a ``fundamental system'' \cite{van2025synchronous} can be lifted to a left-invariant system on a (typically, large) Lie group, facilitating the design of observers whose errors evolve independently of the state or inputs and enabling almost-global stability \cite{Mahony2008,2018zlotnikGradientbasedObserverSimultaneous,2021berkaneNonlinearNavigationObserver,vangoor2025INS,liu2025global}.
Among these, particularly well-studied are the group affine systems on Lie groups \cite{Markus1981Controllability, AyalaTirao1999Linear,Barrau2017}, of which left-invariant systems are a special case.
Although group affine systems need not be  left-invariant nor right-invariant in general, they nonetheless enjoy state-independent error dynamics, enhancing the stability and consistency of many filters \cite{Barrau2017}.
While equivariant, fundamental, and group affine systems each enjoy certain favorable properties reminiscent of classically invariant systems, they all rely on a transitive symmetry on the state space, and therefore do not in fact generalize the classical notion of invariance, where the dimension of the group may be smaller than that of the state space.
Moreover, although left-invariant systems on Lie groups are a special case of both group affine systems and of classically  invariant systems on general manifolds (see Fig.~\ref{fig:venn_diagram}), a more unified notion of symmetry---simultaneously capturing the structure enjoyed by 
both these classes of systems---has been absent.

Only a few authors have proposed notions of symmetry that are strictly more general than classical invariance.
In parallel to early work on classical symmetry, the very general notion of ``partial symmetry'' \cite{nijmeijer1985partial} was proposed to study control systems in which the non-symmetric part of the dynamics is tangent to the orbits of the symmetry group. 
Both classically invariant and partially invariant systems admit a cascade decomposition in which the driven subsystem evolves on the symmetry group; however, in the case of classical symmetry, the subsystem on the symmetry group is left-invariant \cite{Grizzle1985}, whereas for partial symmetry, it has no special structure whatsoever \cite{nijmeijer1985partial}.
Much more recently, \cite{Wang2022ApproximateEquivariance} proposed ``approximate equivariance'', where the failure of the system to be classically invariant is small in a certain sense.
Both these approaches considerably generalize classical invariance, but, in doing so, they sacrifice many algorithmically useful properties that invariance provides.

In this work, we introduce a novel relaxed notion of symmetry that we term ``weak invariance'', which is a strict generalization of classical (henceforth, ``strong'') invariance.
This weak notion of symmetry 
provides many useful structural features reminiscent of strong invariance, and it moreover generalizes group affine systems to non-transitive symmetries.
We begin by introducing a ``residual'' term that captures the non-symmetric part of a given system's dynamics under a given group action. 
This residual affords a convenient characterization of the existing notions of strong symmetry (for which it vanishes completely) and of partial symmetry (for which it is ``vertical'', \textit{i.e.}, tangent to the group orbits). 
Similarly, we define weak symmetry as the situation in which the residual evaluates to infinitesimal generators of the group action, and we show that a system is weakly invariant if and only if the residual can be expressed entirely in terms of another control system evolving on the symmetry group.

Our remaining contributions characterize the rich structure 
that weak symmetry grants a control system.
First, a weakly invariant system admits a cascade decomposition in which the subsystem on the symmetry group is group affine, establishing a middle ground between the existing decomposition results for strong \cite{Grizzle1985} and partial \cite{nijmeijer1985partial} symmetry. 
Second, the group affine structure of these dynamics
 can be exploited to show that the group error dynamics in a weakly invariant system with autonomous residual enjoy a state-independence property analogous to that studied in the invariant filtering literature \cite{Barrau2017,vanGoor2021}.
In effect, such a property amounts to factoring out the symmetry group from the joint dynamics governing the relative motion of two trajectories, thereby reducing its total dimension by that of the group.
Third, a subgroup of a weak symmetry is another weak symmetry if and only if that subgroup is preserved by the control system induced on the original symmetry group by the residual; this condition can also be used to extract the largest strong symmetry from any given weak symmetry. 
Finally, as an example, we consider an aerial vehicle under the symmetry-breaking effects of gravity, which enjoys a well-known four-dimensional strong symmetry corresponding to translation and rotation around the vertical axis. 
We show that the aerial vehicle also admits a larger nine-dimensional weak symmetry, which we factor out of the error dynamics to achieve a greater reduction in dimensionality.
Considering the significant demonstrated benefits of cascade decompositions and state-independent error dynamics across observer design 
\cite{Mahony2008,Barrau2017,vangoor2025INS},
control design and analysis \cite{Welde2024, Welde2023b}, and efficient learning of tracking controllers \cite{welde2025leveragingsymmetryacceleratelearning,pagnini_error}, we believe  the rich structure of weakly invariant systems will enable new methods of symmetry-informed estimation, learning, and control for a broader class of control systems.

This paper is an evolution of our previous conference publication \cite{weldevanGoor2026weaknotionsymmetrydynamical}, in which we proposed the notion of weak symmetry for dynamical systems (lacking inputs) but did not consider control systems, error dynamics, or subgroups of weak symmetries.
The organization of the paper proceeds as follows. In Sec.~\ref{sec:preliminaries}, we recall certain essential aspects of differential geometry and Lie group theory. In Sec.~\ref{sec:notions_of_symmetry}, we introduce the residual, use it to define weak symmetry (and characterize strong and partial symmetry), and discuss technical relationships to other previously-studied notions of symmetry. In Sec.~\ref{sec:system_structure}, we study the implications of increasingly rigid notions of symmetry on the structure of a control system, obtaining a novel decomposition for weak symmetry and sharpening existing results for strong and partial symmetries. In Sec.~\ref{sec:error_dynamics}, we characterize the error dynamics of a weakly invariant system with autonomous residual, and in Sec.~\ref{sec:symmetry_subgroups}, we study subgroups of weak symmetries. 
Finally, Sec.~\ref{sec:aerial_vehicle} considers an aerial vehicle under the influence of gravity, identifying a novel weak symmetry,  
recovering a familiar strong symmetry as a subgroup thereof, and formulating error dynamics reduced by the weak symmetry. 
We  conclude the paper in Sec.~\ref{sec:conclusion}.

\section{Mathematical Preliminaries}
\label{sec:preliminaries}

For a thorough introduction to smooth manifolds, the authors recommend \cite{SmoothLee}.
Consider smooth ($C^\infty$) manifolds $M,N$.
The set of all diffeomorphisms ${f : M \to M}$ is denoted $\Diff(M)$.
The tangent space at each point ${x \in M}$ is denoted as ${\T_x M}$ and the tangent bundle is denoted $\T M$.
For any map ${h : M \to N}$ and any ${y \in N}$, ${h \equiv y}$ means that ${h(x) = y}$ for all ${x \in M}$. 
If $h$ is injective, its left-inverse ${h^{-1} : h(M) \to M}$ is well-defined and satisfies ${h^{-1} \circ h = \id_M}$.
If $h$ is differentiable, then the \textit{differential} of $h$ is a map between tangent bundles denoted ${\diff h : \T M \to \T N}$.
Given a map with two arguments ${f : M_1 \times M_2 \to N}$, we will often use the partial map notation ${f_x = f(x, \, \cdot\, ) : M_2 \to N}$ for any fixed ${x \in M_1}$, and likewise ${f^y = f(\, \cdot\, , y) : M_1 \to N}$ for any fixed ${y \in M_2}$. $\X(M)$ denotes the set of vector fields on $M$ (\textit{i.e.}, all smooth maps ${v: M \to \T M}$ such that ${v(x) \in \T_x M}$ for all ${x \in M}$).

\subsection{Lie Groups and Automorphisms}

A Lie group $G$ is a smooth manifold equipped with a smooth group structure.
The product of any two elements ${g,h \in G}$ is written as ${gh \in G}$, and the identity element is denoted ${e \in G}$.
The left and right translations ${\Left,\Right : G \times G \to G}$ are defined by ${\Left_g(h) := gh}$ and ${\Right_g(h) := hg}$, respectively.
A vector field ${v \in \X(G)}$ is \textit{left-invariant} if  ${\diff \Left_{g^{-1}} \circ v \circ \Left_g = v}$ for all ${g \in G}$, and the set of all left-invariant vector fields is denoted ${\mathfrak{left}(G) \subset \X(G)}$. 
The Lie algebra ${\mathfrak{g} = \T_e G}$ is identified with the tangent space of the group at the identity element $e$. 
The Lie bracket in $\mathfrak{g}$ is induced by its isomorphism with the set of left-invariant vector fields $\mathfrak{left}(G)$.

For a Lie group $G$, a diffeomorphism ${\sigma : G \to G}$ is called an \textit{automorphism}
if ${\sigma(g_1)\sigma(g_2) = \sigma(g_1 g_2)}$ for all ${g_1, g_2 \in G}$.
The set of all automorphisms of $G$ is denoted ${\Aut(G) \subset \Diff(G)}$.
When $G$ is connected, $\Aut(G)$ is itself a Lie group of dimension at most $(\dim G)^2$ \cite{Hochschild1952_automorphism}.
A vector field ${w \in \mathfrak{X}(G)}$ is called \textit{group linear} if 
\begin{align}
    w(gh) = \diff \Left_g \circ w(h) + \diff \Right_h \circ w(g),
    \label{definition_of_group_linear_vector_field}
\end{align}
for all ${g, h \in G}$.
The set of all group linear vector fields is written as ${\aut(G) \subset \mathfrak{X}(G)}$.
As suggested by the notation, when $\Aut(G)$ is a Lie group, its Lie algebra is  $\aut(G)$.

For any $g \in G$, the conjugation map ${\Cn_g : G \to G}$ defined by ${\Cn_g(h) = g h g^{-1}}$ is an automorphism.
Differentiating $\Cn_g(h)$ with respect to  $h$ at the identity ${h = e}$ yields the familiar adjoint representation 
$\Ad_g : \mathfrak{g} \to \mathfrak{g}$, given by
\begin{align}
    \Ad_g(\xi) = \diff(\Cn_g)(\xi) = \diff\Left_g\circ\diff\Right_{g^{-1}}(\xi).
\end{align}
Conversely, differentiating $\Cn_g(h)$ with respect to $g$ at the identity ${g = e}$ yields a map ${\cn : \mathfrak{g} \to \aut(G)}$ defined such that for any $\xi \in \mathfrak{g}$, the vector field $\cn_\xi \in \aut(G)$ is given by
\begin{align}
    \cn_\xi(h)
    = \diff (\Cn^h) (\xi)
     = \diff \Right_h(\xi) - \diff \Left_h (\xi).
     \label{definition_of_cn_vector_field}
\end{align}

\subsection{Group Actions}

A (left\footnote{A right action  is defined similarly, but with ${\Phi_g \circ \Phi_h = \Phi_{hg}}$.})
action of a Lie group $G$ on a manifold $M$ is a smooth map
${\Phi : G \times M \to M}$ such that for all $g, h \in G$,
\begin{align*}
    \Phi_g \circ \Phi_h = \Phi_{gh} \quad \textrm{and} \quad \Phi_e = \id_M.
\end{align*}
The \textit{orbit} of any point ${x \in M}$ is given by ${\Phi(G,x) = }$ ${\{\Phi_g(x) : g\in G\}}$. 
The quotient space $M/G$ is the set of all orbits of $\Phi$, and the map ${\pi : M \to M/G}$ sends each point in $M$ to its orbit.
The \textit{infinitesimal generator} of $\Phi$ is the map ${(\,\cdot\,)_M : \mathfrak{g} \to \mathfrak{X}(M)}$ defined such that
\begin{equation}
    \xi_M(x) = 
    \diff \Phi^x (\xi),
\end{equation}
for all ${x \in M}$, and the set of all such generators (also called the \textit{fundamental vector fields} associated with $\Phi$) is given by
\begin{equation}
    \mathfrak{g}(M) = \{\xi_M : \xi \in \mathfrak{g}\}.
\end{equation}
Finally, $\Phi$ is \textit{free} if $\Phi_g$ has no fixed points for all ${g \neq e}$ (\textit{i.e.}, if there exists ${x \in M}$ such that ${\Phi_g(x) = x}$, then ${g = e}$), and $\Phi$ is  \textit{proper} if the map ${(g,x) \mapsto \big(\Phi_g(x),x\big)}$ is a proper map (\textit{i.e.}, the preimage of any compact set is compact).

\subsection{Principal $G$-Bundles and Principal Connections}
\label{subsec:principal_bundles}

When $\Phi$ is free and proper, the quotient space $M/G$ inherits a unique smooth manifold structure such that $\pi$ is a surjective submersion, in which case ${\pi : M \to M/G}$ can be viewed as a principal $G$-bundle. A \textit{(global) section} is a smooth map ${s : M/G \to M}$ such that ${\pi \circ s = \id_{M/G}}$. Any such section 
(when it exists) 
induces a \textit{(global) trivialization}, decomposing $M$ into its base space $M/G$ and the fiber $G$, so that every point $x \in M$ can be uniquely and smoothly decomposed as ${x = \Phi_g \circ s(y)}$, where ${y = \pi(x)}$.

The principal bundle structure induces a \textit{vertical distribution} ${\V M := \ker \diff \pi \subseteq \T M}$ containing all vectors in $\T M$ tangent to the orbits. The set of \textit{vertical vector fields} is given by
\begin{equation}
    \mathfrak{vert}(M) = \big\{ f \in \X(M) : 
    f(x) \in \V_x M \ 
    \textrm{for all}
    \ x \in M
    \big\}.
\end{equation}
At each point $x \in M$, the map ${\diff\Phi^x : \mathfrak{g} \to \V_x M}$ is a bijection.
Moreover, one may choose any smooth fiber-preserving map ${\hor : \T M \to \T M}$ such that:
\begin{itemize}
    \item $\hor$ is idempotent (\textit{i.e.}, $\hor\circ\hor = \hor$),
    \item $\hor$ is equivariant (\textit{i.e.}, $\hor \circ \diff\Phi_g = \diff\Phi_g\circ\hor$), and
    \item $\hor$ annihilates only $\V M$ (\textit{i.e.}, $\ker (\hor) = \V M$).
\end{itemize}
Then ${\H M := \hor(\T M)}$ is called the \textit{horizontal distribution}, such that ${\T M = \H M \oplus \V M}$.
Any tangent vector can then be split uniquely as
${v_x = \hor(v_x)+\ver(v_x)},$
where the vertical projection is ${\ver := \id_{TM} - \, \hor}$. 
Conversely, a \textit{principal connection one-form} is a map  ${\mathcal{A} : \T M \to \mathfrak{g}}$  such that
\begin{itemize}
    \item $\mathcal{A}$ is equivariant (\textit{i.e.}, ${\mathcal{A}\circ \diff\Phi_g = \Ad_g \circ \, \mathcal{A}}$ for all ${g \in G}$), 
    \item $\mathcal{A}$ maps vertical tangent vectors to their generators (\textit{i.e.}, ${\mathcal{A}\big(\xi_M(x)\big) = \xi}$ for all ${x \in M}$), and
    \item $\mathcal{A}$ annihilates $\H M$ (\textit{i.e.}, $\mathcal{A}(\H M) = 0 \in {\mathfrak{g}}$).
\end{itemize}
Principal connection one-forms and horizontal projections are in one-to-one correspondence. In particular, we may define the principal connection ${\mathcal{A} : \T M \to \mathfrak{g}}$ corresponding to any valid choice of $\hor : \T M \to \T M$ by pulling the vertical projection back to the Lie algebra, namely,
\begin{equation}
    \mathcal{A}(v_x) := \big(\diff \Phi^x\big)^{-1} \circ \ver(v_x),
    \label{define_connection_one_form}
\end{equation}
since ${\diff\Phi^x : \mathfrak{g} \to \V_x  M}$ is a bijection.

In the remainder of the paper, we will choose a horizontal projection induced by a choice of section ${s: M/G \to M}$ (although our global results can easily be adapted to local results using local sections). 
In particular, the \textit{canonical flat connection}
(making $\H M$  integrable) 
induced by a chosen section $s$ is defined so that for each point ${x = \Phi_g \circ s(y)}$, 
\begin{align}
    \label{eq:define_vertical_projection_flat_connection}
    \hor_x(v_x) &:= \diff \Phi_g \circ \diff s \circ \diff \pi(v_x),    
\end{align}
and since $\ver (v_x) = v_x - \hor(v_x)$, we have
\begin{equation}
    \label{eq:define_horizontal_projection_flat_connection}
    \ver_x(v_x) = v_x - \diff \Phi_g \circ \diff s \circ \diff \pi(v_x).
\end{equation}
Thus, from \eqref{define_connection_one_form}, the connection one-form at ${x \in M}$ is given  by
\begin{align}
    \label{eq:define_principal_connection_one_form_flat_connection}
    \mathcal{A}_x(v_x) &= \big(\diff \Phi^x\big)^{-1} \big( v_x - \diff \Phi_g \circ \diff s \circ \diff \pi(v_x) \big).
\end{align}

\subsection{Control Systems}

A \textit{control system} with state space $M$ and input space $\mathbb{U}$ is modeled as a smooth\footnote{${f : \mathbb{U} \to \X(M)}$ is smooth if ${(x,u) \mapsto f_u(x)}$ is smooth (\textit{i.e.}, ${C}^\infty$).} map ${f : \mathbb{U} \to \X(M)}$. For each ${u \in \mathbb{U}}$, we write ${f_u := f(u) \in \X(M)}$.
A \textit{solution} of $f$ consists of a smooth \textit{trajectory} ${x : \mathbb{R} \to M}$ and a smooth \textit{input signal} ${u : \mathbb{R} \to \mathbb{U}}$ such that for all $t \in \mathbb{R}$, 
\begin{equation}
    \dot{x}(t) = f_{u(t)}\big(x(t)\big).
\end{equation}
A control system ${f : \mathbb{U} \to \X(M)}$ is said to be \textit{autonomous} if ${f_{u_1} = f_{u_2}}$ for all ${u_1, u_2 \in \mathbb{U}}$.

For any Lie group $G$, 
consider a control system ${f : \mathbb{U} \to \X(G)}$.
This system is \textit{left-invariant} if for all ${u \in \mathbb{U}}$, 
we have ${f_u \in \mathfrak{left}(G)}$.
More generally, $f$ is \textit{group affine} if for all $u \in \mathbb{U}$ and all $g, h \in G$, we have
\begin{align*}
    f_u(gh) = \diff \Left_g \circ f_u(h) + \diff \Right_h \circ f_u(g) - \diff \Left_g \circ \diff \Right_h \circ f_u(e).
\end{align*}
Alternatively, the following fact provides a more structural characterization of group affine control systems.

\begin{fact}[{\textit{cf.} {\cite[Thm. 4.3]{vanGoor2021}}}]
\label{fact:group_affine_group_linear_decomposition}
    A control system ${f : \mathbb{U} \to \X(G)}$ is group affine if and only if 
    there exist (unique) control systems 
    \begin{align}
        {w : \mathbb{U} \to \aut(G) }
        ,  
        &&
        {v : \mathbb{U} \to \mathfrak{left}(G)}
    \end{align}
    such that ${f_u = w_u + v_u}$ for all ${u \in \mathbb{U}}$.
    In particular, 
    \begin{align}
        w_u(g) = f_u(g) - \diff \Left_g  \circ f_u(e),
        &&
        v_u(g) = \diff \Left_g \circ f_u(e),
        \label{eq:left_invariant_part_determined_by_evaluation_at_identity}
    \end{align}
    for all $g \in G$ and $u\in \mathbb{U}$.
\end{fact}

\section{Symmetries of Control Systems}
\label{sec:notions_of_symmetry}

Throughout this paper, we consider a control system ${f : \mathbb{U} \to \mathfrak{X}(M)}$ and a smooth Lie group action ${\Phi : G \times M \to M}$.
We make the following standing assumption, which will remain in force for the rest of the paper.

\begin{assumption}
    ${\Phi : G \times M \to M}$ is a free and proper action.
\end{assumption}

The following definition recalls a classical notion of symmetry and $\Phi$-invariance for control systems.

\begin{definition}[{\textit{cf.} \cite{Grizzle1985}}]
The control system ${f : \mathbb{U} \to \X(M)}$ is  \textit{(strongly) $\Phi$-invariant} if for all $g \in G$, $x \in M$, and $u \in \mathbb{U}$,
\begin{align}
    \label{eq:phi_invariance_commutative_square}
    f_u \circ \Phi_g(x) = \diff \Phi_g \circ f_u(x).
\end{align}
In this case, $\Phi$ is said to be a \textit{(strong) symmetry} of $f$.
\end{definition}

We refer to this property as ``strong'' invariance to clearly distinguish it from the weaker properties to be studied.
We postpone a more thorough discussion of nomenclature as it relates to other flavors of symmetry until Section~\ref{sec:symmetry_nomenclature}. 

An equivalent definition of strong invariance is obtained by applying $\diff \Phi_{g^{-1}}$  to both sides of \eqref{eq:phi_invariance_commutative_square} and collecting terms.
That is, the system $f$ is strongly $\Phi$-invariant if and only if
\begin{align*}
    \diff \Phi_{g^{-1}} \circ f_u \circ \Phi_g - f_u = 0.
\end{align*}
This perspective provides a way to study relaxed notions of symmetry via the failure of this quantity to vanish identically.
To this end, we introduce the following \textit{residual}, which 
captures the ``non-symmetric part'' of the dynamics and thus
vanishes precisely when $\Phi$ is a strong symmetry of $f$.

\begin{definition}\label{dfn:residual}
    The \textit{residual} of ${f : \mathbb{U} \to \X(M)}$ with respect to ${\Phi : G \times M \to M}$ is the map
${\Delta :  G \times \mathbb{U} \to \X(M)}$ given by
\begin{equation}
        \Delta{(g,u)} = \diff \Phi_{g^{-1}} \circ f_u \circ \Phi_g - f_u.
        \label{eq:residual_definition}
\end{equation}
The residual is \textit{autonomous} if there exists ${\Delta : G \to \X(M)}$ such that $\Delta(g) = \Delta(g,u)$ 
for all $g \in G$ and ${u \in \mathbb{U}}$. We often use the convenient shorthand $\Delta_{(g,u)} := \Delta(g,u) \in \X(M)$.
\end{definition}

The residual of \textit{any} control system enjoys the following ``one-cocycle'' property, which will be important in the sequel.

\begin{lemma}
\label{lemma:residual_one_cocycle}
    For all ${g, h \in G}$ and ${u \in \mathbb{U}}$, we have
    $$\Delta_{(gh,u)} = \Delta_{(h,u)} + \diff \Phi_{h^{-1}} \circ \Delta_{(g,u)} \circ \Phi_h.$$
\end{lemma}
\begin{proof}
By direct computation, we have
\begin{align}
    &\diff \Phi_{h^{-1}} \circ \Delta_{(g,u)} \circ \Phi_h
    \nonumber
    \\
    &= 
    \diff \Phi_{h^{-1}} \circ \big(
    \diff \Phi_{g^{-1}} \circ f_u \circ \Phi_g - f_u
    \big) \circ \Phi_h
    \nonumber
    \\
    &=
    \diff\Phi_{(gh)^{-1}} \circ f_u \circ \Phi_{gh} - f_u 
    +
    f_u - \diff\Phi_{h^{-1}} \circ f_u
    \circ \Phi_h 
    \nonumber
    \\
    &=
    \Delta_{(gh,u)}
    -
    \Delta_{(h,u)}.
    \qedhere
\end{align}
\end{proof}

The residual allows for the following convenient description of a relaxed notion of symmetry, based on a property first studied in \cite{nijmeijer1985partial}. 

\begin{definition}[{\textit{cf.} \cite{nijmeijer1985partial}}]
\label{defn:res_ps}
The control system ${f : \mathbb{U} \to \X(M)}$ is \textit{partially $\Phi$-invariant} if for all ${g \in G}$ and ${u \in \mathbb{U}}$,
\begin{equation}
    { \Delta_{(g,u)} \in \mathfrak{vert}(M)}.
\end{equation}
In this case, we say that $\Phi$ is a \textit{partial symmetry} of $f$.
\end{definition}

In the following definition, we introduce a different relaxed notion of symmetry, which is a novel contribution of the present work. 
It is also naturally stated in terms of the residual.

\begin{definition}
\label{def:weak_invar_residual}
The control system ${f : \mathbb{U} \to \X(M)}$ is \textit{weakly $\Phi$-invariant} if  for all ${g \in G}$ and ${u \in \mathbb{U}}$,
\begin{equation}
    {\Delta_{(g,u)} \in \mathfrak{g}(M)}.
\end{equation}
In this case, we say that $\Phi$ is a \textit{weak symmetry} of $f$.
\end{definition}

Interestingly, the preceding notions of symmetry establish a filtration of Lie subalgebras of $\X(M)$, given by
\begin{equation*}
    \{0 \} \subseteq \mathfrak{g}(M) \subseteq \mathfrak{vert}(M) \subseteq \X(M),
\end{equation*}
where each subalgebra above respectively describes the allowed residuals under strong symmetry, weak symmetry, partial symmetry, and no symmetry at all.
This clearly shows that weak invariance occupies a middle ground between the previously-studied notions of strong and partial symmetry.

\subsection{A Control System on the Symmetry Group}

While the above definition of weak invariance via the residual is conceptually satisfying, the following proposition establishes a more explicit characterization of weak invariance.

\begin{proposition}\label{prop:weak-invariant_system}
A control system ${f : \mathbb{U} \to \X(M)}$ is weakly $\Phi$-invariant if and only if there exists a control system ${w :\mathbb{U} \to \mathfrak{X}(G)}$ such that for all $g \in G$, $x\in M$, and $u \in \mathbb{U}$,
\begin{align}
    f_u \circ \Phi_g(x) = 
    \diff \Phi\big(w_u(g),f_u(x)\big).
    \label{weak_invariance_of_a_control_system}
\end{align}
In particular, $w : \mathbb{U} \to \X(G)$ is then given by
    \begin{equation}
        w_u(g) := \diff \Left_g \circ \xi_{(g,u)},
        \label{eq:infinitesimal_generator_group_linear_conversion}
    \end{equation}
    where ${\xi : G \times \mathbb{U} \to \mathfrak{g}}$ is the (unique) smooth map defined such that ${\Delta_{(g,u)} = (\xi_{(g,u)})_M}$ for all ${g \in G}$ and ${u \in \mathbb{U}}$.
\end{proposition}

\begin{proof}
        To show sufficiency, we assume that \eqref{weak_invariance_of_a_control_system} holds and act on both sides via $\diff \Phi_{g^{-1}}$. Rearranging terms,
        we obtain
    \begin{align*}
        \diff \Phi_{g^{-1}} \circ f_u \circ \Phi_g(x) - f_u(x)
        &= \diff \Phi_{g^{-1}} \circ \diff \Phi^x \circ w_u(g), 
    \end{align*}
    and thus the residual is given by
    \begin{align}
        \Delta_{(g,u)}(x) = \diff \Phi^{x} \circ \diff \Left_{g^{-1}} \circ w_u(g).
        \label{eq:definition_of_w_from_residual}
    \end{align} 
    Recalling that the infinitesimal generator is given by ${\eta_M(x) = \diff \Phi^x(\eta)}$ for any ${\eta\in\mathfrak{g}}$, and noting that clearly ${\diff \Left_{g^{-1}} \circ w_u(g) \in \mathfrak{g}}$,  it thus follows that ${\Delta_{(g,u)} \in \mathfrak{g}(M)}$.
    
    To show necessity, note that the  weak invariance of $f$ implies the existence of such a map ${\xi : G \times \mathbb{U} \to \mathfrak{g}}$ such that    
    \begin{equation}
    \diff \Phi_{g^{-1}} \circ f_u \circ \Phi_g - f_u
        = (\xi_{(g,u)})_M
    \end{equation}
    for all $g \in G$ and all $u \in \mathbb{U}$.
    Moreover, the uniqueness of $\xi$ follows from the effectiveness of $\Phi$ (since ${(\cdot)_M : \mathfrak{g} \to \X(M)}$ is an injective map). 
    Acting on both sides of the previous equation via $\diff \Phi_g$ and evaluating at any point $x \in M$ yields
    \begin{equation*}
        f_u \circ \Phi_g(x) - \diff \Phi_g \circ f_u(x)  = \diff \Phi_g \circ \diff \Phi^x \circ \xi_{(g,u)}.
    \end{equation*}
    Differentiating the identity ${\Phi_g \circ  \Phi^{x} =  \Phi^x \circ  \Left_g}$ and rearranging terms, we obtain
    \begin{align*}
        f_u \circ \Phi_g(x)  &= \diff \Phi_g \circ f_u (x) + \diff \Phi^{x} \circ \diff \Left_g \circ \xi_{(g,u)}
        \\
        &= \diff \Phi\big(  \diff \Left_g \circ \xi_{(g,u)}, f_u (x)\big).
    \end{align*}
    To complete the argument, it thus suffices to note that \eqref{eq:infinitesimal_generator_group_linear_conversion} is a well-defined control system.
\end{proof}

In view of the last result, it is natural to say that a given control system is weakly invariant, in particular, with respect to another control system evolving on the symmetry group.
We now characterize the properties of this other control system.

\begin{proposition}
\label{prop:properties_of_symmetry_group_control_system}
    Suppose the control system ${f : \mathbb{U} \to \X(M)}$ is weakly $\Phi$-invariant with respect to ${w : \mathbb{U} \to \X(G)}$. Then:
    \begin{enumerate}[(i)]
        \item For any ${u \in \mathbb{U}}$, the vector field ${w_u \in \X(G)}$ is group linear, {\normalfont i.e.}, ${w_u \in \aut(G)}$. \label{item:w_in_aut_G}
        \item 
        ${{w}}$ is autonomous if and only if $\Delta$ is autonomous. 
        \item ${w \equiv 0 \in \X(G)}$ if and only if $f$ is strongly $\Phi$-invariant.
    \end{enumerate}
\end{proposition}

\begin{proof}
Let ${x \in M}$, ${u \in \mathbb{U}}$, and ${g, h \in G}$ be arbitrary.
Then, 
it follows from \eqref{weak_invariance_of_a_control_system} that
\begin{align*}
    \diff \Phi^x \circ w_u(gh) 
    &= f_u \circ \Phi_{gh}(x) - \diff \Phi_{gh} \circ f_u(x) \\
    &= f_u \circ \Phi_{g} \big( \Phi_h (x)\big) - \diff \Phi_{g} \circ \diff \Phi_h \circ f_u(x) \\
    &= \diff \Phi_g \circ f_u \big(\Phi_h (x) \big) + \diff \Phi^{\Phi_h (x) } \circ w_u(g)
    \\ &\hspace{1cm}
    - \diff \Phi_{g} \circ \big( f_u \circ \Phi_h(x) - \diff \Phi^x \circ w_u(h) \big) \\
    &= \diff \Phi^{\Phi_h (x) } \circ w_u(g)  + \diff \Phi_{g} \circ \diff \Phi^x \circ w_u(h) \\
    &= \diff \Phi^{x} \circ \big(  \diff \Right_h \circ w_u(g)  + \diff \Left_g \circ w_u(h) \big).
\end{align*}
Since $\diff \Phi^x : \T G \to \T M$ is injective for each $x \in M$, we have
\begin{align*}
    w_u(gh) = \diff \Right_h \circ w_u(g)  + \diff \Left_g \circ w_u(h),
\end{align*}
for all $g, h\in G$ and $u \in \mathbb{U}$, and therefore $w_u \in \aut(G)$ (proving the first item).
From Prop.~\ref{prop:weak-invariant_system} and \eqref{eq:infinitesimal_generator_group_linear_conversion}, it is clear that
        \begin{align}
        \Delta_{(g,u)} = \big(\diff \Left_{g^{-1}} \circ w_u(g)\big)_M.
        \label{eq:residual_as_infinitesimal_generator_with_W}
    \end{align}
The map ${\diff \Left_{g^{-1}} : \T G \to \T G}$ restricts to a linear isomorphism ${\T_g G \to \mathfrak{g}}$, and 
${(\cdot)_M : \mathfrak{g} \to \X(M)}$ restricts to a linear isomorphism ${\mathfrak{g} \to \mathfrak{g}(M)}$ (since $\Phi$ is effective). Thus,
${\Delta_{(g,u)} \in \mathfrak{g}(M)}$ is independent of $u$ if and only if ${w_u(g) \in \T_g G}$ is independent of $u$ (proving the second item), and
${\Delta_{(g,u)}}$ vanishes if and only if ${w_u(g)}$ vanishes (proving the third item).
\end{proof}

Thus, the control system ${w : \mathbb{U} \to \aut(G)}$ (with respect to which a given system is weakly $\Phi$-invariant) is in fact highly structured, which will play a key role in the results to come.

\subsection{Connections to Other Flavors of Symmetry}
\label{sec:symmetry_nomenclature}

This section discusses technical relationships between the notions of symmetry considered in this work and those studied in prior literature. We  also take the opportunity to disambiguate certain overloaded terminology for symmetry. 

This paper uses the term ``invariance'', rather than ``equivariance'', for its  consistency with the standard notion of a left-invariant vector field on a Lie group. 
The name comes from the fact that a strongly invariant vector field is a fixed (\textit{i.e.}, invariant) point of the (right) $G$-action ${\Upsilon : G \times \X(M) \to \X(M)}$ on the (infinite-dimensional) vector space $\X(M)$ given by ${\Upsilon : (g,f) \mapsto \diff \Phi_{g^{-1}} \circ f \circ \Phi_g}$. 
In contrast, a standard \textit{equivariant} map \cite[Ch.~7]{SmoothLee} is one where two distinct group actions act on the domain and codomain.
In other words, for each $u \in \mathbb{U}$, the map ${f_u : M \to \T M}$ may be called equivariant if there exist group actions ${\Phi : G \times M \to M}$ and ${\Psi : G \times \T M \to \T M}$ such that ${f_u \circ \Phi_g = \Psi_g \circ f_u}$ for all ${g \in G}$.
Since the map ${  \Psi : (g,v_x) \mapsto \diff \Phi_g(v_x)}$ is itself a $G$-action on $\T M$, strong invariance is a (very) special case of the more general property of equivariance for the map ${f_u : M \to \T M}$ (for each ${u \in \mathbb{U}}$).
However, the following result shows that very little structure is imposed by requiring the map ${f_u}$ to be equivariant in general.

\begin{proposition}\label{prop:many_systems_are_equivariant}
    Suppose ${f : \mathbb{U} \to \X(M)}$ has an autonomous residual $\Delta$
     with respect to ${\Phi : G \times M \to M}$, and define the map ${\Psi : G \times \T M \to \T M}$ given by%
    \begin{align}
        \Psi_g(v_x) = \diff \Phi_g \big(v_x + \Delta_{g} (x) \big).
        \label{eq:definition_of_general_equivariant_group_action}
    \end{align}
    Then $\Psi$ is a group action, and ${f_u \circ \Phi_g = \Psi_g \circ f_u}$ for all ${g \in G}$ and ${u \in \mathbb{U}}$. 
\end{proposition}

\begin{proof}
    From \eqref{eq:residual_definition},
    \begin{align*}
        \diff \Phi_g \big(v_x + \Delta_{(g,u)} (x) \big)
        = 
        \diff \Phi_g \big(v_x &- f_u(x)\big) + f_u \circ \Phi_g(x).
    \end{align*}
    Thus, choosing any arbitrary (fixed) ${u \in \mathbb{U}}$, $\Psi$ is clearly smooth (since $\Phi$ and $f$ are as well). Moreover,
    \begin{align*}
        \Psi_e (v_x) 
        = v_x - f_u(x) + f_u(x) = v_x,
    \end{align*}
    and for any $g, h \in G$, we compute
    \begin{align*}
        \Psi_g &\circ \Psi_h (v_x)
        \\ &= 
        \diff \Phi_g \big(
        \diff \Phi_h \big(v_x - f_u(x)\big) + f_u \circ \Phi_h(x)
        - f_u \circ \Phi_h(x) \big) 
        \\ & \quad \quad + f_u \circ \Phi_g \circ \Phi_h(x)
        \\ 
        &= 
        \diff \Phi_{gh} \big(v_x - f_u(x)\big)  + f_u \circ \Phi_{gh}(x) = \Psi_{gh}(v_x).
    \end{align*}
    It remains only to verify directly that
    \begin{equation*}
        \Psi_g \circ f_u(x) = \diff \Phi_g \big(f_u(x) - f_u(x)\big) + f_u \circ \Phi_g(x) = f_u \circ \Phi_g(x),
    \end{equation*}
    so ${f_u : M \to \T M}$ is 
    equivariant with respect to $\Phi$ and $\Psi$.
\end{proof}

Other literature has explored alternative notions of 
symmetry
that consider group actions on the \textit{input} space as well.
Recent work \cite{mahony2022ObserverDesignNonlinear} considers systems with transitive (rather than free and proper) group actions on the state space, defining equivariance of the system as equivariance of the map ${f : \mathbb{U} \to \X(M)}$ itself, where another $G$-action acts on the input space $\mathbb{U}$.
Early work on both strong \cite{Grizzle1985} and partial \cite{nijmeijer1985partial} symmetry (for free and proper group actions) also originally considered an action on the inputs.
However, they showed that if the projection ${\pi : M \to M/G}$ admits a global (resp. local) section, any system that is equivariant in such a sense is globally (resp. locally) feedback equivalent to another control system for which the action on the inputs is trivial (and can thus be omitted).
Since the present work omits an action on the inputs, and Prop.~\ref{prop:many_systems_are_equivariant} shows that even many weakly invariant systems are (perfectly) equivariant in another sense, the proposed notion of symmetry is thus termed ``weak invariance'' to better distinguish it from those other existing notions of ``equivariance''.

Finally, the definition of weak invariance in this paper considers control systems of the form ${f : \mathbb{U} \to \X(M)}$, whereas both strong \cite{Grizzle1985} and partial \cite{nijmeijer1985partial} symmetry were first studied within Brockett's more general ``fiber bundle picture'' of control systems \cite{brockett1977control}, in which the available inputs at a given state constitute the fiber over that state within another fiber bundle. 
When this other bundle is trivial, working instead with a map ${f : \mathbb{U} \to \X(M)}$ incurs no additional loss of generality. 
Thus, our notions of symmetry are equally general (modulo feedback equivalence) whenever the control inputs can be expressed in a state-independent manner and ${\pi : M \to M/G}$ is trivial.\footnote{Of course, local analogs of our results can also be readily obtained in the more general setting, since fiber bundles are locally trivial by definition.}

\section{Symmetry and System Structure}
\label{sec:system_structure}

This section characterizes the implications of increasingly rigid notions of symmetry on the structure of a control system ${f : \mathbb{U} \to \X(M)}$, by gradually strengthening the flavor of symmetry exhibited by $f$ with respect to ${\Phi : G \times M \to M}$.

\subsection{Dynamics in a Trivialization}

Given a global section ${s : M/G\to M}$ of the principal bundle ${\pi : M \to M/G}$, it is illustrative to analyze the dynamics of the control system ${f : \mathbb{U} \to \X(M)}$ in the corresponding global trivialization ${M \simeq M/G \times G}$.
In doing so, the system's evolution is broken down in terms of two interconnected control systems evolving in $G$ and  $M/G$ respectively.
After deriving the general form of this decomposition (even in the total absence of symmetry), these results will be refined under gradually stronger symmetry assumptions.

\begin{lemma}[System Decomposition]\label{thm:feedback_interconnection}
Let ${s : M/G \to M}$ be a global section of ${\pi : M \to M/G}$ and let ${x(t) \in M}$ be a solution of ${f : \mathbb{U} \to \X(M)}$ with input ${u(t) \in \mathbb{U}}$. 
Then, the decomposed states ${y(t) \in M/G}$ and ${g(t) \in G}$, defined by ${x = \Phi_g \circ s(y)}$, are governed by the feedback interconnection
\begin{subequations}
\begin{align}
    \dot y  = f^{M\hspace{-1pt}/G}_{(g,u)}(y) &:= \diff\pi\circ (f_u+\Delta_{(g,u)})\circ s(y),
    \label{eq:feedback_interconnection_ydot_partial}
    \\
    \dot g = f^{\hspace{.5pt}G}_{(y,u)}(g) &:= \diff\Left_g\circ\mathcal{A}\circ (f_u+\Delta_{(g,u)})\circ s(y),
    \label{eq:feedback_interconnection_gdot_partial}
\end{align}
\end{subequations}
where ${\mathcal{A}: \T M \to \mathfrak{g}} $ is the principal connection one-form of the canonical flat connection induced by the section $s$.
\end{lemma}

\begin{proof}
Substituting the decomposition ${{x} = \Phi_g \circ s(y)}$ into the dynamics ${\dot{x} = f_u(x)}$ and applying the chain rule, we have
    \begin{align}
        \diff\Phi^{s(y)}(\dot g) + \diff\Phi_g\circ \diff s(\dot y) \label{eq:cascade_decomp_nosym}
        &= 
        f_u \circ \Phi_g \circ s(y).
    \end{align}
    Projecting the previous equation to $\T(M/G)$ via ${\diff \pi}$ annihilates the first term on the left 
    (since it is tangent to the orbits), while the second term projects to $\dot{y}$, yielding
        \begin{align}
    \dot y &= 
            \diff \pi \circ f_u \circ \Phi_g \circ s(y).
    \label{eqn:verification_of_y_dot_partial}
    \end{align}
    To obtain \eqref{eq:feedback_interconnection_ydot_partial}, it thus suffices to observe 
    from \eqref{eq:residual_definition}
    that
    \begin{align}
        \diff \Phi_{g^{-1}} \circ f_u \circ \Phi_g = f_u + \Delta_{(g,u)} \label{eq:rewrite_pullback_using_residual}
    \end{align}
    and recall that ${\pi \circ \Phi_{g^{-1}} = \pi}$.
    To extract the group dynamics, we 
    insert \eqref{eqn:verification_of_y_dot_partial} back into \eqref{eq:cascade_decomp_nosym} and rearrange terms, yielding
    \begin{align*}
                \diff\Phi^{s(y)}(\dot g)  
        &= 
        f_u \circ \Phi_g \circ s(y) - \diff\Phi_g\circ \diff s \circ 
                \diff\pi \big(  f_u \circ \Phi_g \circ s(y)\big)
        \\ 
        &= \ver \big( f_u \circ \Phi_g \circ s(y) \big),
    \end{align*}
    where the second equality follows from
    \eqref{eq:define_vertical_projection_flat_connection}. %
    Acting on both sides via $\diff \Phi_{g^{-1}}$, we use the identities ${\Phi_h \circ \Phi^q = \Phi^q \circ \Left_h}$ and 
    ${\diff \Phi_h \circ \ver = \ver \circ \diff \Phi_h}$
    to obtain
    \begin{align*}
        \diff\Phi^{s(y)} \big( \diff \Left_{g^{-1}}(\dot g)  \big)
        &= \ver \big( \diff \Phi_{g^{-1}} \circ f_u \circ \Phi_g \circ s(y) \big).
    \end{align*}
    In view of 
    \eqref{eq:define_principal_connection_one_form_flat_connection} %
    and \eqref{eq:rewrite_pullback_using_residual}, we conclude that
    \begin{align*}
        \diff \Left_{g^{-1}}(\dot g)  
        &= \mathcal{A} \circ \big( f_u + \Delta_{(g,u)} \big)  \circ s(y),
    \end{align*}
    and so the result follows by applying $\diff \Left_g$ to both sides.
\end{proof}

Although no genuine structural simplification of the dynamics has yet been achieved (as expected, given the absence of any symmetry assumptions), this technical result isolates the role of the residual in the decomposed dynamics of arbitrary control systems, facilitating the refinement of this decomposition under gradually stronger symmetry assumptions.

\subsection{Cascade Decompositions}
\label{sec:cascade_decompositions}

\begin{figure}[t]
    \centering

    {\textsf{\textsc{{Partial Symmetry}}}}%
    \vspace{6pt}
    
    \begin{tikzpicture}

\filldraw[thick, fill=ieee_blue!10, draw=ieee_blue!50, very thin]
    (2.6,-.7) 
    rectangle
    (-.8,1);
    
    \node[yshift=.25cm] (top_of_f) at (2.7*.5-.9*.5,1) { 
        \footnotesize 
        \textsf{control system on $M/G$}
        };

\filldraw[thick, fill=ToolOrange!8, draw=ToolOrange!50,very  thin]
    (6.8,-.7) 
    rectangle
    (3.5,1);

    \node[yshift=.25cm] (top_of_g) at (6.8*.5 + 3.5*.5 ,1) { 
        \footnotesize 
        \textsf{control system on $G$}
        };

\node[block] (f) {$\tilde{f}$};
\node[integrator, right=.8cm of f] (inty) {$\int$};

\node[block, right=3.5cm of f, ] (v) {$f^G$};

\node[integrator, right=.8cm of v] (intg) {$\int$};

\draw[arrow]
    ($(f.west)+(-1.35,0)$)
    -- node[above] {$u$} 
    ($(f.west)+(-.8,0)$)
    coordinate(usplice)
    node[splice]
    |- (f.west);

\draw[arrow]
	(usplice)
    |- 
    ($(v)+(0,-1)$)
    coordinate(uspliceagain)
    --
    (v.south);

\draw[arrow]
    (f.east)
    -- node[above] {$\dot y$} (inty.west);

\draw[arrow]
    (inty.east)
    -- 
        node[above] {$y$}
    ++(.5,0)
    coordinate (ysplice)
    node [splice]
    -- ++(0,0.65)
    -- ($(f) + (0cm, 0.65cm)$)
	-- (f.north)
    ;

\draw[arrow]
    (ysplice)
    --
        node[above, pos=.5] {$y$}
    (v.west);

\draw[arrow]
    (v.east)
    -- node[above, pos=.4] {$\dot g$} (intg.west);

\draw[arrow]
    (intg.east)
    -- 
	node[above] {$g$}
	++(.5,0)
    coordinate (gout)
    |- ($(v)+(0,0.65)$)
    -- (v.north);

\end{tikzpicture}
    \caption{A partially $\Phi$-invariant system admits a cascade decomposition in which ${y \in M/G}$ evolves independently of ${g \in G}$. However, in general, the subsystem dynamics
     ${\tilde{f} : \mathbb{U} \to \X(M/G)}$ and ${f^{\hspace{.5pt}G} : M/G \times \mathbb{U} \to \X(G)}$ enjoy no special structure of their own.}
    \label{fig:block_diagram_partial_symmetry}
\end{figure}
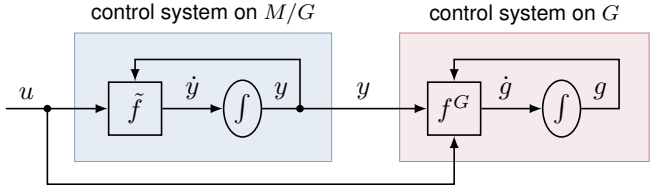

It was first shown in \cite{nijmeijer1985partial} that when $\Phi$ is a partial symmetry, 
the feedback interconnection 
\eqref{eq:feedback_interconnection_ydot_partial}-\eqref{eq:feedback_interconnection_gdot_partial}
degenerates into a cascade in which the dynamics on
$M/G$
evolve independently. In the following theorem, we sharpen that result, showing that partial symmetry is also \textit{necessary} for such a decomposition.

\begin{theorem}[Partial Symmetry, {\textit{cf.}}{\cite{nijmeijer1985partial}}]\label{thm:partial_cascade}
Let all assumptions of Lemma~\ref{thm:feedback_interconnection} hold. Then, 
$\Phi$ is a partial symmetry of $f$
if and only if
 \eqref{eq:feedback_interconnection_ydot_partial}-\eqref{eq:feedback_interconnection_gdot_partial} is a cascade in which 
 \eqref{eq:feedback_interconnection_ydot_partial}
 evolves independently.
In particular, 
the subsystems then evolve according to%
\begin{subequations}
\begin{align}
    \dot{y} &= \tilde{f}_u(y), \label{eq:cascade_ydot_partial}\\
   \dot{g} &= \diff\Left_g\circ\mathcal{A}\circ (f_u+\Delta_{(g,u)})\circ s(y), \label{eq:cascade_gdot_partial}
\end{align}
\end{subequations}
where ${\tilde{f}: \mathbb{U} \to \mathfrak{X}(M/G)}$ is the unique control system 
defined (independent of the choice of $s$) such that for all $u \in \mathbb{U}$,
    \begin{align}
        \label{define_projected_control_system}
        \tilde{f}_u \circ \pi = \diff \pi \circ f_u.
    \end{align}
\end{theorem}
\begin{proof}
For sufficiency, assume that 
$\Phi$ is a partial symmetry of $f$. 
To show that \eqref{define_projected_control_system} defines a unique map ${\tilde{f} : \mathbb{U} \to \X(M/G)}$, 
it suffices to check that 
${\diff \pi \circ f_u(x_1) = \diff \pi \circ f_u(x_2)}$ whenever ${x_1 = \Phi_g (x_2)}$ for some ${g \in G}$. We compute
\begin{align*}
    \diff \pi \circ f_u(x_1)&=
    \diff \pi \circ f_u \circ \Phi_g (x_2)
    \\&=    \diff \pi \circ \diff \Phi_{g^{-1}} \circ f_u \circ \Phi_g (x_2)
    \\&=    \diff \pi \circ (f_u + \Delta_{(g,u)}) (x_2)
    =    \diff \pi \circ f_u (x_2),
\end{align*}
    since $\Phi$ is a partial symmetry of $f$ by assumption, and thus ${\Delta_{(g,u)} \in \mathfrak{vert}(M)}$ by definition. 
To verify that $\tilde{f}$ is smooth, it suffices to apply both sides of \eqref{define_projected_control_system} to $s(y)$, yielding
\begin{equation}
    {\tilde{f}_u(y) = \diff \pi \circ f_u \circ s(y)}.
\label{eq:explicit_f_tilde_in_terms_of_section}
\end{equation}
Since $\Delta$ takes values in $\mathfrak{vert}(M)$,  this last result and \eqref{eq:feedback_interconnection_ydot_partial} imply that $\dot{y} = \tilde{f}_u(y)$, as required.

For necessity, suppose that the feedback interconnection \eqref{eq:feedback_interconnection_ydot_partial}-\eqref{eq:feedback_interconnection_gdot_partial} degenerates into a cascade in which \eqref{eq:feedback_interconnection_ydot_partial} evolves independently.
Then, for any point ${q \in s(M/G)}$,  we have
\begin{align}
    \diff \pi \circ \Delta_{(g,u)}(q) =     \diff \pi \circ \Delta_{(e,u)}(q) =
    0_{\pi(q)},
\label{residual_is_vertical_along_the_section}
\end{align}
for all ${g \in G}$ and  ${u \in \mathbb{U}}$, 
since ${\Delta_{(e,u)} = f_u - f_u = 0}$. It remains to show that 
the left-most side of \eqref{residual_is_vertical_along_the_section}
also vanishes even when ${q \not\in s(M/G)}$. 
Noting that for any ${x \in M}$, there exists ${h \in G}$ and ${q \in s(M/G)}$ such that ${x = \Phi_h(q)}$, 
we 
compute
\begin{align*}
    \diff \pi  \circ \Delta_{(g,u)}(x) 
    &=
    \diff \pi \big( \diff \Phi_{h^{-1}} \circ \Delta_{(g,u)} \circ \Phi_h(q)\big)
    \\
    &=
    \diff \pi \big(
    \Delta_{(gh,u)}(q)
    -
    \Delta_{(h,u)}(q)  \big)
    \\
    &=
    \diff \pi \circ 
    \Delta_{(gh,u)}(q)
    -
    \diff \pi \circ \Delta_{(h,u)}
    (q) 
    = 0_{\pi(x)},
\end{align*}
where the second equality follows from Lemma~\ref{lemma:residual_one_cocycle} and 
the third equality follows from \eqref{residual_is_vertical_along_the_section}.
Since $x$, $g$, and $u$ were taken to be arbitrary, 
it follows that 
${\Delta_{(g,u)} \in \mathfrak{vert}(M)}$, 
as required.
\end{proof}

Thus, partial symmetry (or lack thereof) completely determines whether or not a system admits a cascade decomposition of the form shown in Fig.~\ref{fig:block_diagram_partial_symmetry}.
While a cascade structure can itself be useful for design and analysis, the above decomposition enforces no structure whatsoever in the group dynamics \eqref{eq:cascade_gdot_partial}---in particular, given any (``desired'') control system ${f^{\hspace{.5pt}G} : M/G \times \mathbb{U} \to \X(G)}$, there exists a residual ${\Delta : G \times \mathbb{U} \to \mathfrak{vert}(M)}$ such that $f^{\hspace{.5pt}G}(g)$ equals the right-hand side of \eqref{eq:cascade_gdot_partial}.
In contrast, the following result  (a novel contribution of this work) shows that  a weakly invariant system enjoys a very structured group subsystem $f^{\hspace{.5pt}G}$, which takes the form shown in Fig.~\ref{fig:block_diagram_weak_symmetry}. 

\begin{figure}[t]
    \centering

    {\textsf{\textsc{{Weak Symmetry}}}}%
    \vspace{6pt}
    
    \begin{tikzpicture}

\filldraw[thick, fill=ieee_blue!10, draw=ieee_blue!50, very thin]
    (2.6,-.7) 
    rectangle
    (-.8,1);
    
    \node[yshift=.25cm] (top_of_f) at (2.7*.5-.9*.5,1) { 
        \footnotesize 
        \textsf{control system on $M/G$}
        };

\filldraw[thick, fill=ToolOrange!8, draw=ToolOrange!50,very  thin]
    (7,-3) 
    rectangle
    (3.5,1);
        \node[yshift=.25cm] (top_of_g) at (7*.5 + 3.5*.5 ,1) { 
        \footnotesize 
        \textsf{group affine system on $G$}
        };

\node[block] (f) {$\tilde{f}$};
\node[integrator, right=.8cm of f] (inty) {$\int$};

\node[block, right=3.5cm of f, ] (v) {$v$};
\node[block, right=3.5cm of f, yshift=-2cm] (w) {$w$};

\node[sum, right=.1cm of v, yshift=-1cm] (sum) {};
\node[integrator, right=.4cm of sum] (intg) {$\int$};

\draw[arrow]
    ($(f.west)+(-1.35,0)$)
    -- node[above] {$u$} 
    ($(f.west)+(-.8,0)$)
    coordinate(usplice)
    node[splice]
    |- (f.west);

\draw[arrow]
	(usplice)
    |- 
    ($(v)+(0,-1)$)
    coordinate(uspliceagain)
    node[splice]
    --
    (v.south);
\draw[arrow]
	(uspliceagain)
    --
    (w.north);

\draw[arrow]
    (f.east)
    -- node[above] {$\dot y$} (inty.west);

\draw[arrow]
    (inty.east)
    -- 
        node[above] {$y$}
    ++(.5,0)
    coordinate (ysplice)
    node [splice]
    -- ++(0,0.65)
    -- ($(f) + (0cm, 0.65cm)$)
	-- (f.north)
    ;

\draw[arrow]
    (ysplice)
    --
        node[above, pos=.5] {$y$}
    (v.west);
\draw[arrow]
    (v.east)
    -| (sum.north);

\draw[arrow]
    (w.east)
    -| (sum.south);

\draw[arrow]
    (sum.east)
    -- node[above, pos=.4] {$\dot g$} (intg.west);

\draw[arrow]
    (intg.east)
    -- 
	node[above] {$g$}
	++(.5,0)
    coordinate (gout)
    node [splice]
    |- ($(v)+(0,0.65)$)
    -- (v.north);

\draw[arrow]
    (gout)
    |- ($(w)+(0,-0.65)$)
    -- (w.south);

\end{tikzpicture}
    \vspace{-12pt}
    \caption{A weakly $\Phi$-invariant system admits a cascade decomposition in which the dynamics of ${g \in G}$ consist of a group linear term ${w : \mathbb{U} \to \aut(G)}$ (independent of ${y \in M/G}$) and  a left-invariant term ${v : M/G \times \mathbb{U} \to \mathfrak{left}(G)}$. The dynamics on $G$ are thus group affine.}
    \label{fig:block_diagram_weak_symmetry}
\end{figure}
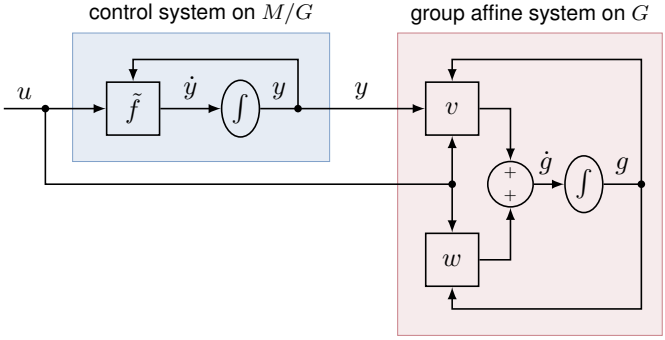

\begin{theorem}[Weak Symmetry]\label{thm:weak_cascade}
    Let all assumptions of Lemma~\ref{thm:feedback_interconnection} hold.
    Then, $\Phi$ is a weak symmetry of $f$ if and only if
$f^{M/G}$ in \eqref{eq:feedback_interconnection_ydot_partial} is independent of $g$ and 
$f^G$ in \eqref{eq:feedback_interconnection_gdot_partial}
is a group affine control system whose group linear term is independent of $y$.
In particular, 
the subsystems then evolve according to
    \begin{subequations}
        \begin{align}
            \dot{y} &= \tilde{f}_u(y), \label{eq:cascade_ydot_weak}\\
            \dot{g} &= w_u(g) + v_{(y,u)}(g),
            \label{eq:cascade_gdot_weak}
        \end{align}
    \end{subequations}
    where 
    $\tilde{f}$
    is defined by \eqref{define_projected_control_system},
    $f$ is weakly $\Phi$-invariant with respect to ${w : \mathbb{U} \to \aut(G)}$,     and 
    ${v : M/G \times \mathbb{U} \to \mathfrak{left}(G)}$ 
    is given by
    \begin{equation}
        {v_{(y,u)}(g) = \diff \Left_g \circ \mathcal{A}\circ f_u \circ s(y)}.
        \label{eq:left_invariant_component_along_orbits}
    \end{equation}
\end{theorem}

\begin{proof}
For necessity, assume that $\Phi$ is a weak symmetry of $f$ with respect to ${w : \mathbb{U}\to \aut(G)}$. Then $\Phi$ is also a partial symmetry of $f$, and Theorem~\ref{thm:partial_cascade} implies that \eqref{eq:feedback_interconnection_ydot_partial}-\eqref{eq:feedback_interconnection_gdot_partial} is a cascade with ${\dot y = \tilde f_u(y)}$ (\textit{i.e.}, \eqref{eq:feedback_interconnection_ydot_partial} is independent of $g$).
Moreover, 
from Prop.~\ref{prop:weak-invariant_system},
there exists a smooth map ${\xi:G\times \mathbb{U}\to \mathfrak{g}}$ such that ${\Delta_{(g,u)} = (\xi_{(g,u)})_M}$. Substituting this into \eqref{eq:feedback_interconnection_gdot_partial}, we compute
\begin{align*}
    \dot{g} 
    &=\diff\Left_g\circ\mathcal{A}\circ \big(f_u+ (\xi_{(g,u)})_M\big)\circ s(y)\\
    &=\diff\Left_g\circ \mathcal{A}\circ f_u \circ s(y) + \diff \Left_g \circ \xi_{(g,u)}
\end{align*}
where the second equality follows from the fact that ${\mathcal{A} \circ \xi_M(x)= \xi}$ for all ${\xi \in \mathfrak{g}}$ and ${x \in M}$ (a defining property of principal connection one-forms).
Defining  $w$ as in \eqref{eq:infinitesimal_generator_group_linear_conversion}  and $v$ as above thus yields precisely \eqref{eq:cascade_gdot_weak}. Since the first term of \eqref{eq:cascade_gdot_weak} is group linear (from Prop.~\ref{prop:properties_of_symmetry_group_control_system}) and
the second term of \eqref{eq:cascade_gdot_weak} is left-invariant, \eqref{eq:feedback_interconnection_gdot_partial} 
is a group affine control system by Fact~\ref{fact:group_affine_group_linear_decomposition} (and in particular, $w_u(g)$ is independent of $y$).

For sufficiency, assume that 
\eqref{eq:feedback_interconnection_ydot_partial} is independent of $g$ and
\eqref{eq:feedback_interconnection_gdot_partial} is a group affine control system whose group linear term is independent of $y$. 
We will show that for all ${g \in G}$ and ${u \in \mathbb{U}}$, we have ${\Delta_{(g,u)} \in \mathfrak{g}(M)}$  (as in Def.~\ref{def:weak_invar_residual}).
By Fact~\ref{fact:group_affine_group_linear_decomposition}, 
$f^{\hspace{.5pt}G}$  takes the form \eqref{eq:cascade_gdot_weak} for some  ${w : \mathbb{U} \to \aut(G)}$
and 
${v : M/G \times \mathbb{U} \to \mathfrak{left}(G)}$.
In particular, evaluating \eqref{eq:feedback_interconnection_gdot_partial} at ${e \in G}$ and applying 
\eqref{eq:left_invariant_part_determined_by_evaluation_at_identity},
we may conclude that
\begin{subequations}
    \begin{align}
    v_{(y,u)}(g) 
    &= \diff \Left_g \circ \mathcal{A}\circ f_u \circ s(y),
    \\ w_u(g) 
    &= \diff\Left_g\circ\mathcal{A}\circ \Delta_{(g,u)} \circ s(y).
\label{eq:group_linear_term_with_connection_and_section}
\end{align}
\end{subequations}
Since \eqref{eq:feedback_interconnection_ydot_partial}-\eqref{eq:feedback_interconnection_gdot_partial} is a cascade (by assumption), Theorem~\ref{thm:partial_cascade} implies that $\Phi$ is (at least) a partial symmetry of $f$, and thus $\Delta$ takes values in $\mathfrak{vert}(M)$. This implies the existence of a smooth map  ${\zeta:M/G\times G\times  \mathbb{U}\to \mathfrak{g}}$ such that
$${\Delta_{(g,u)} \circ s(y) = (\zeta_{(y,g,u)})_M \circ s(y)}$$
for all ${y \in M/G}$, ${g \in G}$, and ${u \in \mathbb{U}}$. 
Inserting this into 
\eqref{eq:group_linear_term_with_connection_and_section}, the properties of principal connection one-forms imply that
\begin{gather}
    w_u(g) = \diff\Left_g\circ \zeta_{(y,g,u)}
    \label{eq:left-trivialized-group linear-term}
\end{gather}
for all  ${y \in M/G}$, ${g \in G}$, and
$u \in \mathbb{U}$. Since the left-hand side of the last equation is independent of $y$, the right-hand side must be also, and there exists a map $\xi : G \times U \to \mathfrak{g}$ such that 
\begin{align}
\Delta_{(g,u)} (q) 
= (\zeta_{(\pi(q),g,u)})_M (q)
= (\xi_{(g,u)})_M (q)
\label{eq:redefine_lie_algebra_representative_of_residual}
\end{align}
for all $q \in s(M/G)$, $g \in G$, and $u \in \mathbb{U}$. It remains to show that this same relation holds even when $q \not \in s(M/G)$. 
Noting that for any ${x \in M}$, there exists ${h \in G}$ and ${q \in s(M/G)}$ such that ${x = \Phi_h(q)}$, we use Lemma \ref{lemma:residual_one_cocycle} to compute
\begin{align*}
     \Delta_{(g,u)}(x) &= 
    \diff \Phi_h \circ \diff \Phi_{h^{-1}}
    \circ  \Delta_{(g,u)} \circ \Phi_h (q)
    \\ &=  \diff \Phi_h \big( \Delta_{(gh,u)} (q) - \Delta_{(h,u)} (q) \big)
    \\ &=  \diff \Phi_h \circ(\xi_{(gh,u)} - \xi_{(h,u)})_M(q)
\end{align*}
Therefore, applying the connection one-form to both sides and recalling that $\mathcal{A}\circ \diff\Phi_h = \Ad_h\circ\, \mathcal{A}$, we compute
\begin{align}
    \mathcal{A}\circ\Delta_{(g,u)}(x) &= \mathcal{A}\circ\diff \Phi_h \circ(\xi_{(gh,u)} - \xi_{(h,u)})_M(q) \nonumber  \\
    &= \Ad_h\circ \, \mathcal{A}\circ(\xi_{(gh,u)} - \xi_{(h,u)})_M(q) \nonumber \\
    &=\Ad_h(\xi_{(gh,u)} - \xi_{(h,u)}).\label{eq:weak_cascade_proof_Ad_xi_gh}
\end{align}
Meanwhile, \eqref{eq:left-trivialized-group linear-term}, \eqref{eq:redefine_lie_algebra_representative_of_residual}, and Fact~\ref{fact:group_affine_group_linear_decomposition} imply that
\begin{align*}
    w : (g,u) \mapsto \diff\Left_{g}\circ\xi_{(g,u)}
\end{align*}
is a well-defined group linear control system. Thus, we have
\begin{align*}
     \diff\Left_{gh}\circ\xi_{(gh,u)}
     &=\diff\Left_{gh}\circ\xi_{(h,u)}+\diff\Right_h\circ\diff\Left_g\circ\xi_{(g,u)},
\end{align*}
from which we may obtain another ``one-cocycle'' identity,
\begin{align}\label{eq:weak_cascade_proof_xi_gh}
    \xi_{(gh,u)} 
   = \xi_{(h,u)} + \Ad_{h^{-1}} \circ \, \xi_{(g,u)}.
\end{align}
Simplifying \eqref{eq:weak_cascade_proof_Ad_xi_gh} using the last equation yields
\begin{align*}
     \mathcal{A}\circ\Delta_{(g,u)}(x) 
     &=\xi_{(g,u)}.
\end{align*}
Having already established that ${\Delta_{(g,u)} \in \mathfrak{vert}(M)}$, it suffices to note that $\mathcal{A}$ maps every vertical vector to its generator, and thus ${\Delta_{(g,u)} =(\xi_{(g,u)})_M}$ for all ${g \in G}$ and ${u \in \mathbb{U}}$.
\end{proof}

It was first shown in \cite{Grizzle1985} that strongly invariant systems admit a cascade decomposition
of the form shown in Fig.~\ref{fig:block_diagram_strong_symmetry}, in which the group dynamics are left invariant.
The previous theorem generalized that conclusion to weakly invariant systems, replacing ``left invariant'' with ``group affine''.
The final theorem in this section recovers (and also sharpens) the prior result for strongly invariant systems as a corollary of our Theorem~\ref{thm:weak_cascade}, showing also that strong symmetry is in fact \textit{necessary} for such a decomposition.

\begin{figure}[t]
    \centering

    \textsf{{\textsc{Strong Symmetry}}}%
    \vspace{6pt}
    
    \begin{tikzpicture}

\filldraw[thick, fill=ieee_blue!10, draw=ieee_blue!50, very thin]
    (2.6,-.7) 
    rectangle
    (-.8,1);
    
    \node[yshift=.25cm] (top_of_f) at (2.7*.5-.9*.5,1) { 
        \footnotesize 
        \textsf{control system on $M/G$}
        };

\filldraw[thick, fill=ToolOrange!8, draw=ToolOrange!50,very  thin]
    (6.8,-.7) 
    rectangle
    (3.5,1);

    \node[yshift=.25cm] (top_of_g) at (6.8*.5 + 3.5*.5 ,1) { 
        \footnotesize 
        \textsf{left-invariant system on $G$}
        };

\node[block] (f) {$\tilde{f}$};
\node[integrator, right=.8cm of f] (inty) {$\int$};

\node[block, right=3.5cm of f, ] (v) {$v$};

\node[integrator, right=.8cm of v] (intg) {$\int$};

\draw[arrow]
    ($(f.west)+(-1.35,0)$)
    -- node[above] {$u$} 
    ($(f.west)+(-.8,0)$)
    coordinate(usplice)
    node[splice]
    |- (f.west);

\draw[arrow]
	(usplice)
    |- 
    ($(v)+(0,-1)$)
    coordinate(uspliceagain)
    --
    (v.south);

\draw[arrow]
    (f.east)
    -- node[above] {$\dot y$} (inty.west);

\draw[arrow]
    (inty.east)
    -- 
        node[above] {$y$}
    ++(.5,0)
    coordinate (ysplice)
    node [splice]
    -- ++(0,0.65)
    -- ($(f) + (0cm, 0.65cm)$)
	-- (f.north)
    ;

\draw[arrow]
    (ysplice)
    --
        node[above, pos=.5] {$y$}
    (v.west);

\draw[arrow]
    (v.east)
    -- node[above, pos=.4] {$\dot g$} (intg.west);

\draw[arrow]
    (intg.east)
    -- 
	node[above] {$g$}
	++(.5,0)
    coordinate (gout)
    |- ($(v)+(0,0.65)$)
    -- (v.north);

\end{tikzpicture}
    \caption{A strongly $\Phi$-invariant system admits a cascade decomposition in which the group dynamics $v \!:\! M/G \times \mathbb{U} {\to} \mathfrak{left}(G)$ are left-invariant.}
    \label{fig:block_diagram_strong_symmetry}
\end{figure}
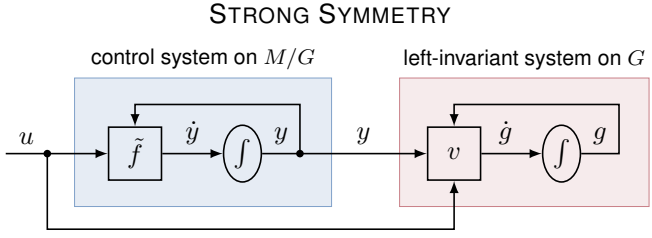

\begin{theorem}[Strong Symmetry, {\textit{cf.}} {\cite{Grizzle1985}}]\label{cor:strong_cascade}
    Let all assumptions of Lemma~\ref{thm:feedback_interconnection} hold. 
    Then, $\Phi$ is a strong symmetry of $f$ if and only if
$f^{M/G}$ in \eqref{eq:feedback_interconnection_ydot_partial} is independent of $g$ and 
$f^{\hspace{1pt}G}$ in \eqref{eq:feedback_interconnection_gdot_partial}
is a left-invariant control system.
In particular, 
the subsystems then evolve according to
    \begin{subequations}
        \begin{align}
            \dot{y} &= \tilde{f}_u(y), \label{eq:cascade_ydot_strong}\\
            \dot{g} &= v_{(y,u)}(g), \label{eq:cascade_gdot_strong}
        \end{align}
    \end{subequations}
    where
        ${\tilde{f}}$  is defined by \eqref{define_projected_control_system} and 
        $v$
        is defined by \eqref{eq:left_invariant_component_along_orbits}.
\end{theorem}

\begin{proof}
Prop.~\ref{prop:properties_of_symmetry_group_control_system} states that a control system is strongly $\Phi$-invariant if and only if it is weakly $\Phi$-invariant with respect to ${w \equiv 0 \in \X(G)}$, and Fact~\ref{fact:group_affine_group_linear_decomposition} implies that a group affine control system is left-invariant if and only if its group linear term ${w}$ vanishes. Thus, the result follows from Theorem~\ref{thm:weak_cascade}.
\end{proof}

\subsection{Control Systems on Lie Groups}

The structure theorems of Section~\ref{sec:cascade_decompositions} characterized control systems evolving on arbitrary manifolds endowed with a free and proper group action. 
The present section examines the special case of a control system on a Lie group $G$ and the (free and proper) left translation action of $G$ on itself.
In this setting, the strongly invariant control systems are precisely the classical left-invariant systems ${f : \mathbb{U} \to \mathfrak{left}(G)}$, which are highly-structured (and thus enjoy many properties useful for control design).
In contrast, restricting partial symmetry to this setting reduces the notion to a vacuous concept, as follows.

\begin{proposition}\label{every_lie_group_system_has_partial_symmetry}
    The left translation action
    ${\Left : G \times G \to G}$ 
    is a partial symmetry of every control system ${f : \mathbb{U} \to \X(G)}$.
\end{proposition}
\begin{proof}
    Since 
    ${\Left : G \times G \to G}$
    is transitive, there is a single orbit (\textit{i.e.}, $G/G$ is a singleton), and thus $\diff \pi$ annihilates every tangent vector in $ \T G$.
    It follows that $\mathfrak{vert}(G) = \X(G)$, hence ${\Delta_{(g,u)} \in \mathfrak{vert}(G)}$ automatically for all ${g \in G}$ and ${u \in \mathbb{U}}$. 
\end{proof}

This straightforward result highlights an essential blindspot of partial symmetry; namely, the residual is completely unconstrained in the vertical directions, and thus no structure is enforced \textit{within} each orbit. In contrast, the much more structured residual of a weakly invariant system (even in the vertical directions) enforces greater rigidity within each orbit, as exemplified by the following interesting corollary.

\begin{corollary}
    A control system ${f : \mathbb{U} \to \X(G)}$ is group affine if and only if it is weakly 
    left-invariant. 
\end{corollary}

\begin{proof}
    This result is an immediate consequence of Theorem~\ref{thm:weak_cascade}, since $G/G$ is a singleton, and thus the only (positive-dimension) subsystem is the one evolving on $G$ itself.
    In particular, the principal connection one-form reduces to the right Maurer-Cartan form, \textit{i.e.}, ${\mathcal{A}(v_g) = \diff \Right_{g^{-1}}(v_g)}$, and a choice of section ${s : G/G \to G}$ amounts to a choice of any (fixed) point in $G$. 
    Evaluating \eqref{eq:cascade_gdot_weak} with ${s \equiv e \in G}$ yields
    \begin{align*}
        \dot{g} 
        = w_u(g) + \diff \Left_g \circ \diff \Right_{{e}^{-1}} \circ f_u(e) 
        = w_u(g) + \diff \Left_g \circ f_u(e),
    \end{align*}
    precisely the form of a group affine control system.
\end{proof}

Thus, weak invariance generalizes not only the classical family of strongly invariant systems, but also the important, well-studied class of group affine systems. 
However, weak symmetry is considerably more general---unlike a group affine system, a weakly invariant system need not evolve on a Lie group, and the group action need not be transitive.

\section{Error Dynamics}
\label{sec:error_dynamics}

This section studies the relative motion of two trajectories of the same control system, a  situation commonly encountered in trajectory tracking control (\textit{i.e.}, actual and reference states) and observer design (\textit{i.e.}, actual and estimated states).

\subsection{Error Dynamics on Lie Groups}
\label{sec:lie_group_error_dynamics}

Consider any control system $f : \mathbb{U} \to \X(G)$ on a Lie group $G$ and any
two trajectories $g(t), \hat{g}(t)$ under respective inputs $u(t), \hat{u}(t)$. Then, the error state ${\tilde{g} = \hat{g}^{-1} g}$ is governed by
\begin{equation*}
    \dot {\tilde{g}} = \ddt (\hat{g}^{-1} g)
        = 
        \diff \Left_{\hat{g}^{-1}} \circ f_u(g)
        -
        \diff \Right_{g} \big(
        \diff \Right_{\hat{g}^{-1}} \circ \, \diff \Left_{\hat{g}^{-1}}  \circ f_{\hat{u}}(\hat{g}) \big),
\end{equation*}
which depends nontrivially on $g$, $\hat{g}$ and $u$, $\hat{u}$ in general.
In contrast, if $f$ is left-invariant, the last equation simplifies to
\begin{align}
    \label{eq:left_invariant_error_dynamics}
    \dot {\tilde{g}} &= 
        f_u(\tilde{g})
        -
        \diff \Right_{\tilde{g}}  \circ f_{\hat{u}}(e),
\end{align}
which depends only on $u$, $\hat{u}$ and the error state $\tilde{g}$ (and \textit{not} on $g$ or $\hat{g}$ individually). Thus, we may effectively ``factor out'' the symmetry group $G$, obtaining a lower-dimensional control system directly governing the error itself. 
This state-independence property 
has been used to great effect in the design of observers for left-invariant systems with favorable convergence properties \cite{Bonnabel2008tac,Bonnabel2009tac}. 
Remarkably, the error dynamics of group affine systems \textit{also} enjoy a similar state independence property, first studied \cite{Barrau2017} in the case that ${u = \hat{u}}$ (as in filtering problems). Here, we prove a slightly different result, replacing the assumption that ${u=\hat {u}}$ with the assumption that the group linear term is autonomous.

\begin{lemma}\label{lem:group_affine_error_dynamics}
Consider a group affine control system with an autonomous group linear term, given by
\begin{align}\label{eq:group_affine_system}
    \dot{g} = w(g) + v_u(g),
\end{align}
 where
${w \in \aut(G)}$
and
${v : \mathbb{U} \to \mathfrak{left}(G)}$. Let
${g(t),\hat{g}(t) \in G}$ be trajectories of this system for inputs ${u(t), \hat{u}(t) \in \mathbb{U}}$ respectively. Then, the  error state ${\tilde{g} = \hat{g}^{-1} g}$ evolves according to
\begin{align}\label{eq:group_affine_error_dynamics}
    \dot{\tilde{g}}
    &= (w - \cn_{v_{\hat{u}}(e)})(\tilde{g}) 
    + (v_u - v_{\hat{u}})(\tilde{g}).
\end{align}
\end{lemma}

\begin{proof}
First, recalling that $w(e) = 0$, one has that
\begin{align*}
    0 &= w(e) = w(g g^{-1})
    = \diff \Right_{g^{-1}} \circ w(g) + \diff \Left_{g} \circ w(g^{-1}),
\end{align*}
and thus $w(g^{-1}) = - \diff \Right_{g^{-1}} \diff \Left_{g^{-1}} w(g)$, for all $g \in G$.
The time derivative of $\hat{g}^{-1}$ is therefore given by
\begin{align*}
    \ddt (\hat{g}^{-1})
    &= -\diff \Right_{\hat{g}^{-1}} \circ \diff \Left_{\hat{g}^{-1}} ( \dot{\hat{g}} ) \\
    &= -\diff \Right_{\hat{g}^{-1}} \circ \diff \Left_{\hat{g}^{-1}} \circ \big( w(\hat{g}) + v_{\hat{u}}(\hat{g})\big) \\
    &= w(\hat{g}^{-1}) - \diff \Right_{\hat{g}^{-1}} \circ v_{\hat{u}}(e).
\end{align*}
From here, the time derivative of $\tilde{g}$ is computed to be
\begin{align*}
    \dot{\tilde{g}}
    &= \diff \Right_g \circ \ddt(\hat{g}^{-1}) +  \diff \Left_{\hat{g}^{-1}} ( \dot{g} ) \\
    &= \diff \Right_g \circ \big(w(\hat{g}^{-1}) - \diff \Right_{\hat{g}^{-1}} \circ v_{\hat{u}}(e)\big) 
    \\ &\hspace{0.5cm}
    +  \diff \Left_{\hat{g}^{-1}} \circ \big(w(g) + \diff \Left_g \circ v_u(e)\big) \\
    &= \diff \Right_g \circ w(\hat{g}^{-1}) - \diff \Right_{\hat{g}^{-1} g} \circ v_{\hat{u}}(e)
    \\ &\hspace{0.5cm}
    +  \diff \Left_{\hat{g}^{-1}} \circ w(g) + \diff \Left_{\hat{g}^{-1} g} \circ v_u(e) \\
    &= \diff \Right_g \circ w(\hat{g}^{-1}) 
    + \diff \Left_{\hat{g}^{-1}} \circ w(g)
    \\ &\hspace{0.5cm}
    - \diff \Right_{\hat{g}^{-1} g} \circ v_{\hat{u}}(e)
    + \diff \Left_{\hat{g}^{-1} g} \circ v_u(e).
\end{align*}
Applying the definition of a group linear vector field \eqref{definition_of_group_linear_vector_field} and that of $\cn$ \eqref{definition_of_cn_vector_field} then gives
\begin{align*}
    \dot{\tilde{g}}
    &= w(\hat{g}^{-1} g) 
    - \cn_{v_{\hat{u}}(e)}(\hat{g}^{-1} g)
    \\ &\hspace{0.5cm}
    - \diff \Left_{\hat{g}^{-1} g} \circ v_{\hat{u}}(e)
    + \diff \Left_{\hat{g}^{-1} g} \circ v_u(e) \\
    &= w(\tilde{g}) 
    - \cn_{v_{\hat{u}}(e)}(\tilde{g})
    - \diff \Left_{\tilde{g}} \circ v_{\hat{u}}(e)
    + \diff \Left_{\tilde{g}} \circ v_u(e).
\end{align*}
This completes the proof.
\end{proof}

Lemma~\ref{lem:group_affine_error_dynamics} shows that the error between two trajectories of a group affine system (with autonomous group linear term) is itself governed by group affine dynamics, though they have a nonautonomous group linear component given by $w - \cn_{v_{\hat{u}}(e)}$.

\subsection{Error Dynamics in Weakly Invariant Systems}

The relative motion of two trajectories of the control system ${f : \mathbb{U} \to \X(M)}$ can be analyzed (in part) by first applying the decomposition of Lemma~\ref{thm:feedback_interconnection} and again studying the evolution of the group error ${\tilde{g} = \hat{g}^{-1} g}$, which now captures only the system's relative motion along the orbits. 
If $\Phi$ is a strong symmetry of $f$, then (by Theorem~\ref{cor:strong_cascade})
the group dynamics are left-invariant, and thus the 
 group \textit{error} dynamics are state-independent (an observation analogous to the reduction perspective explored in \cite{welde2025leveragingsymmetryacceleratelearning}).
 However, if $\Phi$ is only a partial symmetry, the decomposition of Theorem \ref{thm:partial_cascade} recovers no special structure in the group subsystem, and thus no reduction of the group error dynamics can be achieved. 
In contrast, weakly invariant systems \textit{do} admit such a reduction, as follows.

\begin{theorem}[Error Dynamics]\label{thm:error_dynamics}
    Suppose ${f : \mathbb{U} \to \X(M)}$ is weakly $\Phi$-invariant with respect to ${w  \in \X(G)}$ (\textit{i.e.}, $\Delta$ is autonomous).
    Let $x(t),u(t)$ and $\hat x(t), \hat u(t)$ be any two solutions of $f$, where ${x = \Phi_g \circ s(y)}$ and ${\hat x = \Phi_{\hat{g}} \circ s(\hat y)}$ for any global section ${s: M/G \to M}$.
    Then, the error state ${\tilde{g} = \hat{g}^{-1} g \in G}$ and the reduced states 
    ${y, \hat y \in M/G}$
    evolve according to
\begin{subequations}
        \begin{align}
            \dot {\tilde g} &= (
            w - \cn_{v_{(\hat y, \hat u)}(e)}
            )(\tilde{g}) + ({v}_{(y, u)} - {v}_{(\hat y, \hat u)} )(\tilde{g}),\label{eq:error_g}
            \\ \dot y &= \tilde{f}_u(y), \label{eq:error_y}
            \\ \dot {\hat y} &= \tilde{f}_{\hat u}(\hat y), \label{eq:error_y_hat} 
        \end{align}
    \end{subequations}
     where 
    ${{v} : M/G \times \mathbb{U} \to \mathfrak{left}(G)}$ is defined by \eqref{eq:left_invariant_component_along_orbits}
     and
    ${\tilde{f} : \mathbb{U} \to \X(M/G)}$ is defined by \eqref{define_projected_control_system}%
    .
\end{theorem}

\begin{proof}
    From Theorem~\ref{thm:weak_cascade}, the decomposed states ${y, \hat{y} \in M/G}$ and ${g, \hat{g} \in G}$ evolve according to
    \begin{align*}
        \dot{y} &= \tilde{f}_u(y), 
        &
        \dot{\hat{y}} &= \tilde{f}_{\hat{u}}(\hat{y}), 
        \\
        \dot{g} &= w(g) + v_{(y,u)}(g), 
        &
        \dot{\hat{g}} &= w(\hat{g}) + v_{(\hat y, \hat u) }(\hat g),
    \end{align*}
    where the  $g$ and $\hat g$ subsystems are both group affine systems with state space $G$, input space ${M/G \times \mathbb{U}}$, and autonomous group linear term $w$.
    Thus, it suffices to apply
    Lemma~\ref{lem:group_affine_error_dynamics}.
\end{proof}

The dynamics of the error $\tilde{g}$ in \eqref{eq:error_g} of Theorem~\ref{thm:error_dynamics} are themselves group affine, with inputs ${(y, \hat{y}, u,\hat{u}) \in (M/G)^2 \times \mathbb{U}^2}$. In particular, 
the left-invariant term is given by
\begin{align}
    \tilde{v}_{(y, \hat y, u, \hat u)} = 
    (v_{(y,u)} - v_{(\hat y, \hat u)}) \in \mathfrak{left}(G),
\end{align}
while the ({non}autonomous)  group linear term is given by 
\begin{align}
    \tilde{w}_{(\hat y, \hat u)} =  (w - \cn_{v_{(\hat{y}, \hat{u})}(e)}) \in \aut(G).
\end{align}
These dynamics may also be written in the form
\begin{align}
    \dot{\tilde{g}}
    = w(\tilde{g}) 
        + 
    \diff \Left_{\tilde{g}} \circ 
    v_{(y,u)}
    (e)
    - \diff \Right_{\tilde{g}} \circ v_{(\hat{y},\hat{u})}(e) 
    ,
    \label{eq:alternate_form_of_group_error_dynamics}
\end{align}
which emphasizes that the inputs $(\hat{y},\hat{u})$ and $(y,u)$ enter the dynamics of $\tilde{g}$ from ``opposite sides'' but makes the group affine structure less obvious.
When $w \equiv 0$ (\textit{i.e.}, if $\Phi$ is a strong symmetry of $f$), then the dynamics of $\tilde{g}$ simplify further to
\begin{align*}
    \dot{\tilde{g}} = 
    v_{(y,u)}(\tilde{g})
    - \diff \Right_{\tilde{g}} \circ v_{(\hat{y},\hat{u})}(e).
\end{align*}
Thus, in the case of strong symmetry, the dynamics of $g$ and $\hat{g}$ are left-invariant, while the dynamics of $\tilde{g}$ are not.
This contrasts with the case of weak symmetry, where the dynamics of $g$, $\hat{g}$, and $\tilde{g}$ are all group affine.

\section{Symmetry Subgroups}
\label{sec:symmetry_subgroups}

Given a (strong, weak, or partial) symmetry of a system, it is natural to ask about the flavor of symmetry enjoyed by subgroups of that symmetry.
In particular, given any closed subgroup ${H \leq G}$, restricting the first argument of a ${\Phi : G \times M \to M}$ to accept only group elements in $H$ yields a free and proper Lie group action ${\Phi^H : H \times M \to M}$ given by ${\Phi^H(h,x) := \Phi(h,x)}$, which we call a \textit{subgroup of $\Phi$}.
We recall a well-known fact on subgroups of a strong symmetry.

\begin{proposition}
Let ${\Phi : G \times M \to M}$ be a strong symmetry of ${f : \mathbb{U} \to \X(M)}$, 
and fix any closed subgroup ${H \leq G}$. Then, 
$\Phi^H$
is also a strong symmetry of $f$. 
\end{proposition}

\begin{proof}
    Since $H$ is closed, it is a Lie subgroup of $G$, and thus $\Phi^H$ is both a smooth map and a well-defined Lie group action. 
    The result is then immediate, since the residual $\Delta$ with respect to $\Phi$ vanishes identically by assumption, and the residual with respect to $\Phi^H$ is nothing but its restriction ${\Delta |_{H\times\mathbb{U}}}$.
\end{proof}

On the contrary, the last proposition would become false if we replace both mentions of ``strong symmetry'' with ``weak symmetry''---in that case, the residual $\Delta$ need only take values in $\mathfrak{g}(M)$, and $\mathfrak{g}(M)$ is generally not contained within $\mathfrak{h}(M)$ (rather, it is the opposite inclusion that holds). Thus, for a given control system, not every closed subgroup of a weak symmetry yields another weak symmetry (and the same remark holds \textit{a fortiori} for partial symmetries). The following result
resolves this ambiguity, 
completely characterizing whether a closed subgroup of a weak symmetry is also a weak symmetry.

\begin{theorem}\label{prop:subgroups_of_weak_symmetry}

Let ${f : \mathbb{U} \to \X(M)}$ be weakly $\Phi$-invariant with respect to ${w : \mathbb{U} \to \X(G)}$, and 
fix any closed subgroup ${H \leq G}$.
Then, 
$f$ is weakly $\Phi^H$-invariant
if and only if 
${w_u(h) \in \T_h H}$ for all ${h \in H}$ and ${u \in \mathbb{U}}$.
\end{theorem}

\begin{proof}

Since $f$ is weakly $\Phi$-invariant with respect to ${w : \mathbb{U} \to \X(G)}$ and ${H \leq G}$, from \eqref{eq:definition_of_w_from_residual} we have that
\begin{align*}
    \Delta_{(h,u)}(x) = \diff \Phi^{x} \circ \diff \Left_{h^{-1}} \circ w_u(h)
\end{align*}
for all ${x \in M}$,  ${u \in \mathbb{U}}$, and ${h \in H}$.
Moreover, it is clear that
${\Delta_{(h,u)}}$ lies in ${\mathfrak{h}(M)} \leq {\mathfrak{g}(M)}$ if and only if for each 
${h \in H}$ and ${u \in \mathbb{U}}$, 
there exists 
${\xi \in \mathfrak{h} \leq \mathfrak{g}}$ such that for all ${x \in M}$,
\begin{equation*}
    \Delta_{(h,u)}(x) = \diff \Phi^x(\xi).
\end{equation*} 
To complete the argument, it suffices to observe that ${\diff \Left_{h^{-1}} \circ w_u(h) \in \mathfrak{h} \simeq \T_eH}$ if and only if ${w_u(h) \in \T_h H}$.
\end{proof}

This theorem also has a convenient corollary---given a weak symmetry, it is straightforward to extract the largest subgroup thereof that is a \textit{strong} symmetry of the same control system.

\begin{corollary}\label{cor:kernel_to_strong_symmetry}
    Suppose that 
    ${f : \mathbb{U} \to \X(M)}$ is weakly $\Phi$-invariant with respect to ${w : \mathbb{U} \to \X(G)}$. Then, 
    \begin{align}
        K 
        := \big\{k \in G : w_u(k) = 0_k \ \textrm{\normalfont for all} \ u \in \mathbb{U}\big\}
        \label{definition_of_kernel_of_w}
    \end{align}
    is a closed subgroup of $G$. Moreover, $K$ is the largest subgroup of $G$ such that
    $\Phi^K$ is a strong symmetry of $f$.
\end{corollary}
\begin{proof} We first verify that $K$ is in fact a subgroup of $G$.
    Recalling from Prop.~\ref{prop:properties_of_symmetry_group_control_system} that  ${w_u \in \aut(G)}$ for all $u \in \mathbb{U}$, we compute for any ${k_1, k_2 \in K}$ and ${u \in \mathbb{U}}$ that
    \begin{align*}
        w_u(k_1 k_2) &= \diff \Left_{k_1} \circ \underbrace{w_u(k_2)}_{=0_{k_2}} + \,
        \diff \Right_{k_2} \circ \underbrace{w_u(k_1)}_{=0_{k_1}}
        = 0_{k_1k_2}.
    \end{align*}
    Thus, $K$ is preserved by the group operation in $G$. 
    Moreover, 
        \begin{align*}
        w_u(e) = w_u(e e) &= \diff \Left_{e} \circ w_u(e) + 
        \diff \Right_{e} \circ w_u(e) = 2 \,  w_u(e).
    \end{align*}
Thus, $w_u(e) = 0_e$ for all $u \in \mathbb{U}$,
hence $K$ is non-empty. 
    Finally, for any ${k \in K}$, we have
    \begin{align*}
        0_e = 
        w_u(k k^{-1}) &= \diff \Left_{k} \circ w_u(k^{-1}) + \,
        \diff \Right_{k^{-1}} \circ \underbrace{w_u(k)}_{=0_{k}}.
    \end{align*}
    Thus, ${w_u(k^{-1}) = - \diff \Left_{k^{-1}} \circ \diff \Right_{k^{-1}}(0_k) = 0_{k^{-1}}}$, and therefore ${k^{-1} \in K}$, as required. 
    To show that $K$ is closed, we note that $K$ can be expressed in the form
    \begin{equation}
        K = \bigcap_{u \, \in \, \mathbb{U}} K_u, \quad K_u := w_u{}^{-1}\big(\{0_g : g \in G\}\big).
    \end{equation}
    The zero section of any tangent bundle is a closed subset thereof, and the preimage of a closed subset through a continuous map is as well. Thus, $K$ is the intersection of (uncountably many) closed subsets of $G$, and it is thus itself closed.

    By Theorem~\ref{prop:subgroups_of_weak_symmetry},  $f$ is thus weakly {$\Phi^K$-invariant}---in particular, with respect to ${w^K \equiv 0 \in \X(K)}$, since the restriction of ${w : \mathbb{U} \to \X(G)}$ to a control system ${w^K : \mathbb{U} \to \X(K)}$ is well-defined and vanishes by construction. Thus,  by Prop.~\ref{prop:properties_of_symmetry_group_control_system}, $f$ is in fact \textit{strongly} $\Phi^K$-invariant.
    It remains to show that $K$ is the largest such subgroup.
    Consider any other subgroup ${H \leq G}$ for which ${\Phi^H := \Phi |_{H \times M}}$ is a well-defined smooth Lie group action and, moreover, a strong symmetry of $f$. From Prop.~\ref{prop:properties_of_symmetry_group_control_system} and \eqref{eq:residual_as_infinitesimal_generator_with_W}, it follows that for all ${h \in H}$ and ${u \in \mathbb{U}}$, we have
    \begin{equation}
        \big(\diff \Left_{h^{-1}} \circ w_u(h)\big)_M = 0 \in \mathfrak{g}(M).
        \label{eq:vanishing_subgroup_generator}
    \end{equation}
    For each ${h \in H}$,
     the map ${\diff \Left_{h^{-1}} : \T G \to \T G}$ restricts to a linear isomorphism ${\T_h G \to \mathfrak{g}}$. Moreover, since $\Phi$ is effective, the map
${(\cdot)_M : \mathfrak{g} \to \X(M)}$ restricts to a linear isomorphism ${\mathfrak{g} \to \mathfrak{g}(M)}$. Thus, \eqref{eq:vanishing_subgroup_generator} implies that 
    ${w_u(h) = 0_h}$ for all ${u \in \mathbb{U}}$. Hence, ${h \in K}$ and moreover ${H \subseteq K}$, as required.
\end{proof}

We summarize the situation as follows. Any subgroup of a strong symmetry is also a strong symmetry, requiring no further analysis. 
For partial symmetry, checking whether or not a subgroup ${H \leq G}$ also induces a partial symmetry requires checking that $\Delta_{(h,u)}(x)$ is everywhere tangent to the orbits (not only for every ${h \in H}$ and ${u \in \mathbb{U}}$, but also for every ${x \in M}$), necessitating analysis of data defined on the state space. 
In contrast, for a weak symmetry, one may 
determine whether a closed subgroup is also a weak symmetry by simply checking 
that $w_u(h)$ is tangent to $H$ 
for all ${h \in H}$ and ${u \in \mathbb{U}}$. 
It is worth stressing that this procedure for weak symmetries involves the analysis of  data defined on the symmetry group alone and \textit{not} on the state space, unlike for partial symmetry.

\section{Application to an Aerial Vehicle}
\label{sec:aerial_vehicle}

\begin{figure}[!t]

\def\rotorscale{1.3}
    \centering
  \resizebox{\columnwidth}{!}{%
    \begin{tikzpicture}[
      font=\small,
      line cap=round,
      line join=round
    ]

  \coordinate (O) at (.5,1.5);      %
  \coordinate (p) at (6.10,3.75);    %
  \coordinate (q) at (9,2.9);   %

  \coordinate (r30)  at (6.10 + 7.079*1.1-6.10*1.1,3.75 + 2.759*1.1 - 3.75*1.1);
  \coordinate (r90)  at (4.905,2.338);
  \coordinate (r150) at (4.1,3.4);
  \coordinate (r210) at (6.10 + 5.121*1.1 - 6.10*1.1,3.75 + 4.741*1.1 - 3.75*1.1);
  \coordinate (r270) at (7.295,5.162);
  \coordinate (r330) at (8.274,4.171);

    \draw[line width=1.5pt, black!90] (p) -- (q);

    \draw[line width=1.5pt, black!90] (r30) -- (r210);
    \draw[line width=1.5pt, black!90] (r90) -- (r270);
    \draw[line width=1.5pt, black!90] (r150) -- (r330);

  \DrawRotor{r210}{-25}{\rotorscale*.64cm}{\rotorscale*.27cm}
  \DrawRotor{r30} {-25}{\rotorscale*.64cm}{\rotorscale*.27cm}
  \DrawRotor{r270}{ 25}{\rotorscale*.64cm}{\rotorscale*.27cm}
  \DrawRotor{r90} { 25}{\rotorscale*.64cm}{\rotorscale*.27cm}
  \DrawRotor{r150}{ 0}{\rotorscale*.58cm}{\rotorscale*.50cm}
  \DrawRotor{r330}{ 0}{\rotorscale*.58cm}{\rotorscale*.50cm}

\DrawRotorSpin{r30} {-25}{\rotorscale*.64cm}{\rotorscale*.27cm}{-360}{90} %
\DrawRotorSpin{r90} { 25}{\rotorscale*.64cm}{\rotorscale*.27cm}{230}{-290} %
\DrawRotorSpin{r150}{  0}{\rotorscale*.60cm}{\rotorscale*.52cm}{-315}{112} %
\DrawRotorSpin{r210}{-25}{\rotorscale*.64cm}{\rotorscale*.27cm}{300}{-140} %
\DrawRotorSpin{r270}{ 25}{\rotorscale*.64cm}{\rotorscale*.27cm}{-320}{130} %
\DrawRotorSpin{r330}{  0}{\rotorscale*.60cm}{\rotorscale*.52cm}{135}{-291} %

  \begin{pgfonlayer}{front}

    \draw[axis, FrameBlue] (p) -- ++(   0.9596 *1.1,  -0.2813*1.1)
      node[pos=1.05,above=4pt,label backing,text=FrameBlue] {$y_\mathcal{B}$};
\draw[axis, FrameBlue] (p) -- ++(-.775*.85,-.916*.85)
        node[pos=1.2,above=8pt,label backing,text=FrameBlue, xshift=-.05cm] {$x_\mathcal{B}$};
    \draw[axis, FrameBlue] (p) -- ++(0,1)
    node[pos=1,above=2pt,label backing,text=FrameBlue] {$z_\mathcal{B}$};

    \draw[ToolOrange,axis] (q) -- ++(    0.9596 *1.1,  -0.2813*1.1)
      node[anchor=north west,label backing,text=ToolOrange] {$z_\mathcal E$};
    \draw[ToolOrange,axis] (q) -- ++(0,-1.1)
      node[anchor=north east,label backing,text=ToolOrange] {$x_\mathcal E$};
    \draw[ToolOrange,axis] (q) -- ++(.775*.85,.916*.85)
      node[anchor=south east,label backing,text=ToolOrange] {$y_\mathcal E$};

    \ifshowgravity
      \draw[force] (2.,5.30) -- ++(0,-1.85)
        node[pos=.5,right=7pt,label backing] {\large $\mathrm{g}$};
    \fi

    \DrawQuarteredHub{p}
    \node[below=2pt,label backing] at ($(p)+(-.02,-.22)$) {$p$};

    \fill[black] (q) circle[radius=4pt];

    \node[below=5pt,label backing, xshift=-.2cm] at (q) {$q$};

    \ifshowinertialframe
      \fill[black!85] (O) circle[radius=3.5pt, xshift=.01cm, yshift=.01cm];
      \draw[axis, black!85] (O) -- ++(1.25,0)
        node[anchor=west,text=black] {$x_\mathcal I$};
      \draw[axis, black!85] (O) -- ++(.79,.88)
        node[anchor=south west,text=black, yshift=-.2cm] {$y_\mathcal I$};
      \draw[axis, black!85] (O) -- ++(0,1.35)
        node[anchor=south,text=black] {$z_\mathcal I$};
    \fi

    \ifshowequations
      \node[equation box,anchor=north west] at (.22,6.25)
        {$\displaystyle
          \begin{aligned}
            q   &= p + R\,d,\\
            R_q &= R\,S.
          \end{aligned}$};
    \fi
  \end{pgfonlayer}
    \end{tikzpicture}%
  }

  \caption{An aerial vehicle under the influence of gravity (cf. \cite{Rajappa2015} or \cite{Brescianini2016}). The body-fixed frame $\mathcal{B}$ attached at its center of mass has position ${p \in \mathbb{R}^3}$ and orientation ${R\in\mathrm{SO}(3)}$ relative to the inertial frame $\mathcal I$. Another body-fixed frame $\mathcal{E}$ is attached at its end effector, with position ${q = p+Rd \in \mathbb{R}^3}$ and orientation ${R S\in\mathrm{SO}(3)}$ for constant offsets ${d \in \mathbb{R}^3}$  and ${S\in\mathrm{SO}(3)}$.  This system admits a nine-dimensional weak symmetry, which has as a subgroup the familiar four-dimensional strong symmetry corresponding to 
  gravity-preserving rigid displacements.
  }
  \label{fig:aerial-vehicle}
\end{figure}
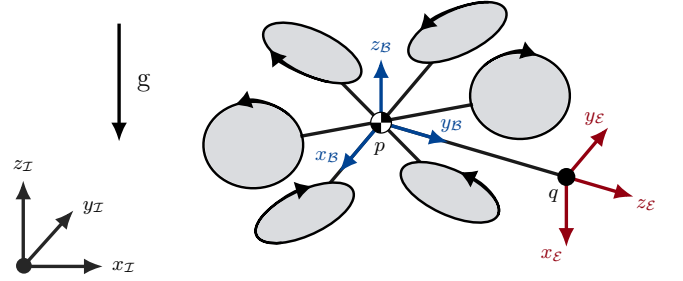

This section studies strong, weak, and partial symmetries of an aerial vehicle \cite{Brescianini2016, Rajappa2015}, consisting of a single rigid body subject to gravity and arbitrary body-fixed control forces and torques  
(see Fig.~\ref{fig:aerial-vehicle}). 
The system can be modeled as a mechanical system with configuration manifold $\mathrm{SE}(3)$ and a left-invariant kinetic energy metric (see \cite{BulloAndLewis2004}), subject to symmetry-breaking potential forces.
The state ${x = (R, p, \nu, \omega) \ \in  \ \mathrm{T}\mathrm{SE}(3)\simeq \mathrm{SO}(3) \times \R^3 \times \R^3 \times \R^3}$ consists of the orientation ${R \in \mathrm{SO}(3)}$, the world-frame position ${p \in \R^3}$ of the center of mass, the body-frame linear velocity ${\nu \in \mathbb{R}^3}$, and the body-frame angular velocity ${\omega \in \mathbb{R}^3}$.
The dynamics  ${f : \mathbb{R}^6 \to \X\big(\mathrm{T}\mathrm{SE}(3)\big)}$ are given by
\begin{align}\label{eq:rigid_body_vf}
    f_u(R, p, \nu, \omega) = \begin{pmatrix}
        R\omega_\times \\[3pt]
        R\nu \\[3pt]
        \tfrac{1}{m} F - \mathrm{g}\, R^\top \mathrm{e}_3 - \omega \times \nu \\[3pt]
        \mathbb{J}^{-1}(\mu - \omega \times \mathbb{J}\omega)
    \end{pmatrix},
\end{align}
where $\omega_\times$ denotes the linear map satisfying ${\omega_\times b = \omega\times b}$ for all ${b\in\R^3}$. 
The control inputs ${u =(\mu, F) \in \R^6 \simeq \R^3 \times \R^3}$ are the body-frame external moment and force, respectively, while ${m > 0}$ is the mass, $\mathbb{J}$ is the moment of inertia tensor, $\mathrm{g}$ is the magnitude of acceleration due to gravity, and ${\mathrm{e}_3 = (0,0,1)}$ is the vertical direction.

\subsection{Symmetries of the Aerial Vehicle}

It is well known (cf. \cite{welde2025leveragingsymmetryacceleratelearning}) that the aerial vehicle described by  \eqref{eq:rigid_body_vf} admits a  strong symmetry corresponding to the four-dimensional subgroup of $\mathrm{SE}(3)$ that preserves the gravity vector. In this section, we show that \eqref{eq:rigid_body_vf} also admits a larger nine-dimensional weak symmetry (of which the strong symmetry is a subgroup) and a twelve-dimensional partial symmetry (of which the weak symmetry is a subgroup). The weak symmetry and the  identification of the strong symmetry as a subgroup thereof are novel contributions of this work.

\subsubsection{Partial Symmetry}
The state manifold $\T \mathrm{SE}(3)$ inherits a Lie group structure induced by the derivative of the usual group operation in $\mathrm{SE}(3)$. In particular, the tangent group operation is given for any group elements $g, h \in \T\mathrm{SE}(3)$ by
\begin{align*}
    &(R_g, p_g, \nu_g, \omega_g) 
    (R_h, p_h, \nu_h, \omega_h) = \\ 
    & \ \ \big(
        R_g R_h, 
        p_g + R_g p_h, 
        R_h^\top (\nu_g + \omega_g \times p_h) + \nu_h, 
        \omega_h + R_h^\top \omega_g
    \big).    
\end{align*}
By Prop.~\ref{every_lie_group_system_has_partial_symmetry}, the left translation action in $\mathrm{TSE}(3)$ is automatically a partial symmetry of \eqref{eq:rigid_body_vf}; however, one may check that it is not a weak (nor, therefore, strong) symmetry. 

\subsubsection{Weak Symmetry} 

The Lie group $\mathrm{SE}_2(3)$, first formalized in the context of inertial navigation in \cite{Barrau2017}, is isomorphic to a subgroup of the tangent group $\mathrm{TSE}(3)$ obtained by setting ${\omega_g = 0}$.
In particular, we express each group element as a tuple ${g = (R_g, p_g, v_g)  \in \mathrm{SE}_2(3)}$, where ${R_g \in \mathrm{SO}(3)}$ and ${p_g, v_g \in \mathbb{R}^3}$, and the group operation is given for any group elements ${g, h \in \mathrm{SE}_2(3)}$ by
\begin{align*}
    (R_g, p_g, v_g)(R_h, p_h, v_h) = (R_g R_h, p_g + R_g p_h, v_g + R_g v_h).
\end{align*}
The identification of $\mathrm{SE}_2(3)$ with a subgroup of $\mathrm{TSE}(3)$ induces an action ${\Phi:\mathrm{SE}_2(3)\times \mathrm{T}\mathrm{SE}(3)\to \mathrm{T}\mathrm{SE}(3)}$ given by 
\begin{align}
    \Phi_g(R, p, \nu, \omega)     &= \left(
        R_g R,\,
        p_g + R_g p,\,
        \nu + R^\top R_g^\top v_g,\,
        \omega
\right), \label{eq:SE23_action}
\end{align}
which is a weak symmetry of the aerial vehicle dynamics \eqref{eq:rigid_body_vf}.
This may be verified by computing the residual and showing it 
is an infinitesimal generator of $\Phi$.
To this end, let ${g \in \mathrm{SE}_2(3)}$ and ${x = (R, p, \nu, \omega) \in \mathrm{TSE}(3)}$ be arbitrary.
Then, evaluating the vector field~\eqref{eq:rigid_body_vf} at the transformed state $\Phi_g(x)$ yields
\begin{align*}
    f_u\big(\Phi_g(x)\big) = \begin{pmatrix}
        R_g R \omega_\times \\[3pt]
        v_g+R_g R\nu \\[3pt]
        \tfrac{1}{m} F - \mathrm{g}\, R^\top R_g^\top \mathrm{e}_3 - \omega \times (\nu + R^\top R_g^\top v_g) \\[3pt]
        \mathbb{J}^{-1}(\mu - \omega \times \mathbb{J}\omega)
    \end{pmatrix}.
\end{align*}
Also, for any
 tangent vector
${\dot x = (\dot R,\dot p,\dot \nu,\dot \omega)\in \mathrm{T}_x\big(\mathrm{TSE}(3)\big)}$, the differential $\diff \Phi_{g^{-1}}$ is given by 
\begin{align*}
    \diff\Phi_{g^{-1}}(\dot x) = \left(
        R_g^\top \dot R,\,
        R_g^\top \dot p,\,
       \dot\nu-\dot R^\top v_g,\,
        \dot \omega
    \right).
\end{align*}
Combining these facts, the residual is computed to be
\begin{align}
    \Delta_{(g,u)}(x)  
    &= \diff \Phi_{g^{-1}} \circ f_u \circ \Phi_g (x) - f_u (x) \notag \\
    &= \left(0,\,R_g^\top v_g,\mathrm{g}R^\top\mathrm{e}_3-\mathrm{g}R^\top R_g^\top \mathrm{e}_3,\,0\right). \label{eq:residual}
\end{align}
To see that this
is in fact an element of $\mathfrak{g}(M)$ for ${\mathfrak{g} = \mathfrak{se}_2(3)}$ and ${M = \mathrm{TSE}(3)}$,
compute the infinitesimal generator of any 
${\xi = (\xi_R,\xi_p,\xi_v)}$  ${\in \mathfrak{se}_2(3)}$ as
\begin{align}\label{eq:infinitesimal_generator}
    \xi_{\raisebox{-2pt}{\scriptsize$\T \mathrm{SE}(3)$}} (x) = \mathrm{d}\Phi^x(\xi) = (\xi_R R, \xi_p+\xi_R p, R^\top \xi_v,0).
\end{align}
Defining $\xi : \mathrm{SE}_2(3) \to \mathfrak{se}_2(3)$ by
\begin{align}\label{eq:lie_algebra_generator}
    \xi(g) = \big(0,\; R_g^\top v_g,\; \mathrm{g}(\mathrm{I}_{3\times 3}-R^\top_g)\mathrm{e}_3\big) \in \mathfrak{se}_2(3),
\end{align}
one may
recover  \eqref{eq:residual} by applying
\eqref{eq:infinitesimal_generator} to \eqref{eq:lie_algebra_generator}, 
verifying that ${\Delta_{(g,u)} = \big(\xi(g)\big)_{\T \mathrm{SE}(3)}}$. 

The residual $\Delta$ given in \eqref{eq:residual} is autonomous (\textit{i.e.}, independent of $u$), and thus the associated $w \in \aut\big(\mathrm{SE}_2(3)\big)$ must also be autonomous (Prop.~\ref{prop:properties_of_symmetry_group_control_system}).
In particular, $w$ is computed as
\begin{align}
    w(R_g,p_g, v_g)
    &= \diff\Left_g\circ\xi(g) 
    = (0, v_g, \mathrm{g}R_g\mathrm{e}_3 - \mathrm{g}\mathrm{e}_3),
    \label{eq:rigid_body_W}
\end{align}
as in Prop.~\ref{prop:weak-invariant_system}.
In summary, this shows that the $\mathrm{SE}_2(3)$-action $\Phi$ given in \eqref{eq:SE23_action} is a weak symmetry of the aerial vehicle dynamics $f$ given in \eqref{eq:rigid_body_vf} with respect to $w$ given in \eqref{eq:rigid_body_W}.

\subsubsection{Strong Symmetry}
Having identified the weak symmetry \eqref{eq:SE23_action}, the largest strong symmetry contained therein can be extracted using Corollary~\ref{cor:kernel_to_strong_symmetry}.
In particular, the closed subgroup defined in \eqref{definition_of_kernel_of_w} specializes to
\begin{align*}
    K &= \{ g \in \mathrm{SE}_2(3) : w(g) = 0\} \\
    &= \{ (R_k,p_k,v_k)\in \mathrm{SE}_2(3) : (0, v_k, \mathrm{g}R_k\mathrm{e}_3 - \mathrm{g}\mathrm{e}_3) = 0_k\}.
\end{align*}
Clearly, ${g \in K}$ if and only if ${v_k = 0}$ and ${R_k\mathrm{e}_3 = \mathrm{e}_3}$. 
This restricts $R_k \in \mathrm{SO}(3)$ to a rotation about $\mathrm{e}_3$, and thus
\begin{equation*}
    K = \left\{ \left( 
\setlength{\arraycolsep}{2pt}
\begin{bmatrix} R_z & 0 \\ 0 & 1 \end{bmatrix}
    , p_k, 0 \right) \in \mathrm{SE}_2(3) : R_z \in \mathrm{SO}(2), \, p_k \in \mathbb{R}^3 \right\}
    \hspace{-2pt} .
\end{equation*}
This shows that $K$ is isomorphic to the semidirect product ${\mathbb{S}^1 \ltimes \mathbb{R}^3}$, and thus we obtain a strong symmetry ${\Theta: ({\mathbb{S}^1 \ltimes \mathbb{R}^3}) \times \mathrm{T}\mathrm{SE}(3)\to \mathrm{T}\mathrm{SE}(3)}$ of the aerial vehicle. 
This is exactly the familiar four-dimensional strong symmetry group of a rigid body under gravitational influence \cite{welde2025leveragingsymmetryacceleratelearning}.

\subsection{Trivialization of the State Manifold}\label{subsec:rigid_body_decomp}

The weak $\Phi$-invariance of \eqref{eq:rigid_body_vf} can be used to derive a number of structural properties, including a cascade decomposition (as in Section~\ref{sec:cascade_decompositions}) and group affine error dynamics (as in Section~\ref{sec:error_dynamics}).
The dependence of these results on a chosen section (see Section~\ref{subsec:principal_bundles}) is described below.

The weak symmetry $\Phi$ can be written as the restriction of left translations on $\mathrm{TSE(3)}$ to a subgroup isomorphic to $\mathrm{SE}_2(3)$, implying that it is both free and proper.
This means that the bundle projection is a well-defined surjective submersion which may be written as
\begin{align*}
    \pi\colon \mathrm{TSE}(3) \to \T \mathrm{SE}(3) / \mathrm{SE}_2(3)
    \simeq \mathbb{R}^3, 
    \  (R,p,\nu,\omega) \mapsto \omega.
\end{align*}
Choosing a global section ${s:\mathbb{R}^3\to \mathrm{TSE}(3)}$ will determine a global trivialization
${x = \Phi_{g} \circ s(\omega) \in \mathrm{TSE(3)}}$, decomposing the state $x$ into the angular velocity ${\omega = \pi(x) \in \mathbb{R}^3}$ and the group element ${g = (R_g,p_g,v_g) \in \mathrm{SE}_2(3)}$, where the physical interpretation of  $g$ is determined by the choice of section. 
To illustrate the impact of this choice, 
consider another body-fixed reference frame 
(\textit{e.g.}, attached to an onboard sensor or end effector, 
as in Fig.~\ref{fig:aerial-vehicle}) 
 shifted relative to the center of mass by a body-frame translation ${d \in \R^3}$ and a rotational offset $S \in \mathrm{SO}(3)$.
We define the section%
\begin{align}
    s(\omega) := (S^\top,-S^\top d,-\omega_\times d, \omega),
\end{align}
such that the trivialization imposes the constraint
\begin{align*}
    (R,p,\nu,\omega) 
    &= (R_gS^\top,p_g-R_gS^\top d, SR_g^\top v_g-\omega_\times d, \omega).
\end{align*}
Solving these equations for the components of $g$ yields
\begin{align}
    R_g = R S, \quad
    p_g = p +R d, \quad
    v_g = R (\nu +\omega \times d),
\end{align}
so that $g$ coincides exactly with the attitude, world-frame position, and world-frame velocity of this other body-fixed frame.
Letting $d = 0$ and $S = \mathrm{I}_{3 \times 3}$, the section becomes
\begin{align}\label{eq:rigid_body_COM_section}
    s_{\text{CoM}}(\omega) = (\mathrm{I}_{3\times 3},0,0,\omega),
\end{align}
and the group states are given by
\begin{align}
      R_g = R, \quad 
      p_g = p, \quad
      v_g = R \nu,
      \label{eq:COM_section_states}
\end{align}
that is,  we recover the vehicle attitude as well as the position and linear velocity of the vehicle's center of mass in world coordinates.
For the remainder, we adopt the section ${s := s_{\text{CoM}}}$ for simplicity of exposition.
However, we emphasize that any other choice of section is equally valid within the framework of weak invariance, providing flexibility to the designer to make a selection suitable for their particular control problem.

\subsection{Cascade Decomposition}

Using the chosen section \eqref{eq:rigid_body_COM_section} and Theorem~\ref{thm:weak_cascade}, we derive a cascade decomposition of the dynamics \eqref{eq:rigid_body_vf}.
First, the dynamics in the quotient space ${\T \mathrm{SE}(3) / \mathrm{SE}_2(3)}$ are given by specializing \eqref{eq:cascade_ydot_weak} and \eqref{define_projected_control_system} to \eqref{eq:rigid_body_vf}, i.e.
\begin{align}\label{eq:rigid_f_tilde}
        \dot{\omega} 
        &=\tilde f_u \circ \pi(x) 
    = \diff\pi \circ f_u(x) 
    = \mathbb{J}^{-1}(\mu - \omega \times \mathbb{J}\omega).
\end{align}
Meanwhile, the group dynamics are given by \eqref{eq:cascade_gdot_weak} with $w$ given as in \eqref{eq:rigid_body_W}.
To compute the left-invariant term, \eqref{eq:define_principal_connection_one_form_flat_connection} is used to evaluate ${\mathcal{A}\colon \T (\T \mathrm{SE}(3)) \to \mathfrak{se}_2(3)}$ on $f_u$ at the section,
\begin{align*}
   \mathcal{A} & \circ f_u \circ s(\omega)
   = \bigl(\diff\Phi^{s(\omega)}\bigr)^{-1}
      \circ \ver \circ f_u\circ s(\omega) 
      \\
   &
   = \bigl(\diff\Phi^{s(\omega)}\bigr)^{-1}
      \bigl(f_u\circ s(\omega)
            -  \diff \Phi_e \circ \diff s\circ \diff\pi\circ f_u\circ s(\omega)\bigr) \\
   &= \bigl(\diff\Phi^{s(\omega)}\bigr)^{-1}
       \bigl(\omega_\times,\,0,\,
            \tfrac{1}{m}F-\mathrm{g} \mathrm{e}_3,\,0\bigr) \\
   &= \bigl(\omega_\times,\,0,\,
            \tfrac{1}{m}F-\mathrm{g} \mathrm{e}_3\bigr)
      \;\in\;\mathfrak{se}_2(3). 
\end{align*}
Thus, the left-invariant term \eqref{eq:left_invariant_component_along_orbits} is given by
\begin{align}
\nonumber
v_{(\omega,u)}(g) &= \diff \Left_g \bigl(\omega_\times,\,0,\,
            \tfrac{1}{m}F-\mathrm{g} \mathrm{e}_3\bigr)
    \\&= \big(R_g\omega_\times,0,R_g(\tfrac{1}{m}F-\mathrm{ge}_3)\big).
    \label{eq:rigid_body_V}
\end{align}
Combining this with the group linear term \eqref{eq:rigid_body_W} yields
\begin{align}
    \dot g 
    &=\overbrace{\big(0, v_g, \mathrm{g}R_g\mathrm{e}_3 - \mathrm{g}\mathrm{e}_3\big)}^{w \, \in \, \aut(\mathrm{SE}_2(3))} + \overbrace{\big(R_g\omega_\times,0,R_g(\tfrac{1}{m}F-\mathrm{ge}_3)\big)}^{v_{(\omega,u)}  \, \in \, \mathfrak{left}(\mathrm{SE}_2(3))} \nonumber \\
    &= (R_g\omega_\times, v_g, \tfrac{1}{m}R_gF-\mathrm{ge}_3). \label{eq:spatial_frame_g_dyn}
\end{align}
Recalling \eqref{eq:COM_section_states} and defining ${v = R \nu}$, the cascade \eqref{eq:cascade_ydot_weak}-\eqref{eq:cascade_gdot_weak} takes the form
\begin{subequations}
\label{eq:cascade_spatial}
\begin{align} 
     \dot \omega &= \mathbb{J}^{-1}(\mu - \omega \times \mathbb{J}\omega),
     \label{eqn:cascade_aerial_vehicle_ydot}
     \\
     (\dot R,\dot p,\dot v) &=(R\omega_\times, v, \tfrac{1}{m}RF-\mathrm{ge}_3). \label{eqn:cascade_aerial_vehicle_gdot}
\end{align} 
\end{subequations}
Thus, the weak symmetry and chosen section yield a decomposition in which Euler's rotation equation governs the independent evolution of the angular velocity ${\omega \in \R^3}$, which cascades into the  ``inertial navigation system'' dynamics (which are famously group affine \cite{Barrau2017}) governing the ``extended pose'' ${g \in \mathrm{SE}_2(3)}$ consisting of the aerial vehicle's position, orientation, and linear velocity. 

We contrast \eqref{eqn:cascade_aerial_vehicle_ydot}-\eqref{eqn:cascade_aerial_vehicle_gdot} with the more traditional decomposition that splits the system into its pose ${(R,p) \in \mathrm{SE}(3)}$ and twist ${(\nu, \omega) \in \mathbb{R}^6 \simeq \mathfrak{se}(3)}$ (as in the Euler-Poincar\'{e} equations \cite{poincare1901_euler_poincare_equation,BulloAndLewis2004}). 
The classical decomposition corresponds (via Lemma~\ref{thm:feedback_interconnection}) to the usual action of $\mathrm{SE}(3)$ on $\mathrm{TSE}(3)$, which leaves the body-frame twist unaltered, and it is a strong symmetry for systems with external forces strictly in the body frame.
However, due to the symmetry-breaking force of gravity, this action of $\mathrm{SE}(3)$ is not even a partial symmetry of \eqref{eq:rigid_body_vf}, and hence the associated decomposition is not a cascade (\textit{i.e.}, the linear acceleration $\dot{\nu}$ depends on the orientation $R$).

\subsection{Error Dynamics}

As shown in Theorem~\ref{thm:error_dynamics}, the weak symmetry $\Phi$ can be used to reduce the error dynamics governing the relative motion of two trajectories ${x(t), \hat{x}(t) \in \mathrm{TSE}(3)}$ for the aerial vehicle under respective control inputs ${u(t), \hat{u}(t) \in \mathbb{R}^6}$.
This can be useful, for example, when $\hat x(t), \hat u(t)$ is a reference trajectory and the goal is to drive the system's state $x(t)$ asymptotically to $\hat{x}(t)$ by choosing $u(t)$. 

Using the section given in \eqref{eq:rigid_body_COM_section} as before, the states $x,\hat{x}$ decompose as $x = \Phi_{(R,p,v)}\circ s(\omega)$ and $\hat{x} = \Phi_{(\hat{R},\hat{p},\hat{v})}\circ s(\hat{\omega})$, respectively.
The error state $\tilde{g} := \hat{g}^{-1} g$  is thus given by
\begin{align}
     \tilde{g} = (\tilde{R}, \tilde{p}, \tilde{v}) 
    &= \bigl(\hat{R}^\top R,\; \hat{R}^\top(p - \hat{p}),\; \hat R^\top(v-\hat v)\bigr),
\end{align}
and its dynamics can be computed either directly or by specializing \eqref{eq:error_g} or \eqref{eq:alternate_form_of_group_error_dynamics} to the present example.
To write the error dynamics explicitly, note from \eqref{eq:rigid_body_V} that
\begin{align*}
    \diff \Left_{\tilde{g}} \circ v_{(\omega,u)}(e) &= 
    (\tilde{R} \omega_\times, 0, \tfrac{1}{m}\tilde{R} F - \mathrm{g}\tilde{R}\mathrm{e}_3),
    \\ \diff \Right_{\tilde{g}} \circ v_{(\hat{\omega},\hat{u})}(e) &= 
    (\hat\omega_\times \tilde{R}, \hat\omega_\times \tilde{p}, \tfrac{1}{m} \hat F - \mathrm{g} \mathrm{e}_3+ \hat\omega_\times \tilde{v}).
\end{align*}
Inserting these, as well as \eqref{eq:rigid_body_W}, into \eqref{eq:alternate_form_of_group_error_dynamics} provides
\begin{subequations}
    \begin{align}
    \dot{\tilde{R}} &=
        \tilde{R} \omega_\times - \hat\omega_\times \tilde{R},
    \\
    \dot{\tilde{p}} &=
    \tilde{v} - \hat\omega_\times \tilde{p},
    \\
    \dot{\tilde{v}} &= 
    \tfrac{1}{m} (\tilde{R} F - \hat{F}) - \hat\omega_\times \tilde{v}
    ,
\end{align}
\label{eq:se_2_3_error_dynamics}
\end{subequations}
and from \eqref{eq:rigid_f_tilde}, 
$\omega$ and $\hat\omega$ are governed by
\begin{subequations}
    \begin{align}
        \dot{\omega} &=  \mathbb{J}^{-1}(\mu - \omega \times \mathbb{J}\omega), \\
    \dot{\hat\omega} &=  \mathbb{J}^{-1}(\hat\mu - \hat\omega \times \mathbb{J}\hat \omega).
\end{align}
\end{subequations}
The key characteristics of this result are that the dynamics of $\tilde{g}$ are group affine, they depend neither on $g$ nor $\hat{g}$ but rather on $\tilde{g}$ itself, and they form a cascade with the $\omega$ and $\hat{\omega}$ dynamics.
The fact that $\dot{\tilde{g}}$ is independent of $g$ and $\hat{g}$ is a key feature of weakly invariant systems, which allows the dimension of the error dynamics to be reduced by factoring out the group.

The state manifold $\mathrm{TSE}(3)$ of the aerial vehicle is 12-dimensional, meaning that two trajectories together are governed by a 24-dimensional system. Factoring out the weak $\mathrm{SE}_2(3)$ symmetry eliminates nine dimensions, yielding a 15-dimensional error system.
In contrast, factoring out only the four-dimensional strong $\mathbb{S}^1\ltimes \R^3$ symmetry would result in a 20-dimensional error system (and it is not possible to factor out a partial symmetry in general).
Thus, studying the weak invariance of this system greatly improves the extent of reduction achieved for the error system, which has many benefits (\textit{e.g.}, in efficient learning of  tracking controllers \cite{welde2025leveragingsymmetryacceleratelearning}).

\section{Conclusion}
\label{sec:conclusion}

In this work, we proposed weak invariance as a new relaxed notion of symmetry that generalizes strong invariance and extends group affine systems to non-transitive symmetries.
Weak invariance occurs when the failure of a system to be strongly invariant is entirely captured by dynamics on the symmetry group, which is shown to have powerful implications on the structure of the system.
We show that a weakly invariant system can be decomposed into a cascade where the driven subsystem is group affine, and, if the group linear component of these dynamics is autonomous, we show how this leads to a group error with state-independent dynamics, yielding a reduction in the dimension of the joint dynamics.
We further demonstrate that a subgroup of a given weak symmetry provides another weak symmetry if and only if that subgroup is closed under the group linear dynamics on the symmetry group, and that a strong symmetry can be recovered from a weak symmetry using this insight.
Using an aerial vehicle under the influence of gravity as an example, we show that a nine-dimensional weak symmetry can be factored out of the error dynamics to achieve a significant reduction in dimensionality.
In summary, control systems exhibiting weak symmetry  
enjoy several key structural characteristics
similar to those enjoyed by strongly invariant systems, 
which have seen numerous impactful applications \cite{Mahony2008,Barrau2017,vangoor2025INS,Welde2024, Welde2023b,welde2025leveragingsymmetryacceleratelearning,pagnini_error}.
The contributions of this paper therefore serve as a basis for the exploration of novel symmetry-based methods in learning, estimation, and control for the new class of weakly invariant  systems.

\bibliographystyle{IEEEtran}
\bibliography{references.bib}

\end{document}